\documentclass[screen,acmsmall]{acmart}

\usepackage[noxcoloropts]{jhupllab}
\usepackage{tikz}
\usepackage[pro]{fontawesome5}

\usepackage{jhupllab}
\usepackage{amsmath}
\usepackage{mathpartir}
\usepackage{accents}
\usepackage{wasysym}
\SetSymbolFont{wasy}{bold}{U}{wasy}{m}{n} 
\usepackage{stmaryrd}
\SetSymbolFont{stmry}{bold}{U}{stmry}{m}{n} 
\usepackage{mathtools}
\usepackage{amsthm}
\usepackage{thmtools}
\usepackage{ifthenx}
\usepackage{relsize}
\usepackage{array}
\usepackage{wrapfig}
\usepackage{xargs}
\usepackage{soul}
\usepackage{bm}
\usepackage{thm-restate}
\usepackage[outline]{contour}

\newcommand{\bbrule}[4][]{\inferrule*[left={\bbrulename{#2}},#1]{#3}{#4}}
\newcommand{\bbrulename}[1]{{\sc #1}}

\newcommand{\fnstyle}[1]{\text{\textsmaller{\sc{#1}}}}
\newcommand{\deffn}[2]{%
    \expandafter\newcommand\expandafter{\csname #1\endcsname}[1][]{\fnstyle{#2}\ifthenelse{\isempty{##1}}{}{(##1)}}%
}

\newcommand{\kwd}[1]{\upshape \color{NavyBlue}{#1}}
\newcommand{\symbl}[1]{\upshape \color{RoyalPurple}{#1}}

\newcommand{\gtcolon}{\texttt{\symbl{:}}}
\defgt[gtquestion]{\texttt{\symbl{?}}}
\newcommand{\gtarrow}{\ensuremath{\mathrel{\bgroup\contourlength{0.02em}\contour{black}{$\shortrightarrow$}\egroup}}}
\defgt[ob]{\symbl{\char`\{}}
\defgt[cb]{\symbl{\char`\}}}
\newcommand{\lbl}{\ell}
\newcommand{\gtin}{\mathrel{\texttt{\kwd{in}}}}
\newcommand{\gtdot}{\texttt{\symbl .}}

\defgt[gtfun]{\kwd{fun}}
\defgt[gttrue]{\kwd{true}}
\defgt[gtfalse]{\kwd{false}}
\defgt[gtlet]{\kwd{let}}
\defgt[gtletassert]{\kwd{letassert}}

\defgt[gtint]{\kwd{int}}
\defgt[gtnot]{\kwd{not}}

\defgn[binop]{\odot}
\defgt[gtplus]{\texttt{\kwd{+}}}
\defgt[gtminus]{\texttt{\kwd{-}}}
\defgt[gtless]{\texttt{\kwd{<}}}
\defgt[gtand]{\texttt{\kwd{and}}}
\defgt[gtor]{\texttt{\kwd{or}}}
\defgt[gtleq]{\texttt{\kwd{<=}}}
\defgt[gteq]{\texttt{\kwd{=}}} 
\defgt[gtxor]{\texttt{\kwd{xor}}}

\newcommand{\reallab}[1]{{\texttt{\color{ForestGreen}#1}}}
\newcommand{\dbrlab}[1]{{\texttt{\color{RoyalPurple}#1}}}

\defgn[expr]{e}
\defgn[frm]{f}
\newcommand{\dbn}{\fnstyle{n}}
\newcommand{\glab}{\ell\mskip -6.5mu \ell}
\defgn[gexpr]{\expr_\textrm{glob}}
\defgn[meexpr]{\expr_\textrm{me}}
\defgn[env]{E}
\defgn[ecl]{c}
\defgn[eval]{v}
\defgn[keval]{k}
\defgn[reval]{r}
\defgn[revalz]{\dot{r}}
\defgn[ev]{x}
\defgn[ebools]{B}
\defgn[eint]{n}
\defgn[ebool]{\beta}
\newcommand{\stub}{\circ}

\newcommand{\stackname}[1]{\fnstyle{StackName}\ifthenelse{\isempty{#1}}{}{(#1)}}
\newcommand{\extract}[2]{\fnstyle{Extract}\ifthenelse{\isempty{#1}}{}{(#1,#2)}}
\newcommand{\rawval}[1]{\fnstyle{RawVal}\ifthenelse{\isempty{#1}}{}{(#1)}}

\newcommand{\steps}{\Rightarrow}
\newcommand{\globs}{\mathrel /}

\newcommand{\stk}{\Sigma}
\newcommand{\stks}{\bm{\Sigma}\mskip -11mu\bm{\Sigma}}
\newcommand{\display}{{\textrm{D}}}

\newcommand{\cycmark}{c}
\newcommand{\envt}{E}

\newcommand{\cachef}{F}
\newcommand{\cachel}{L}
\newcommand{\pathc}{\pi}
\newcommand{\sfrags}{S}
\newcommand{\visited}{V}

\newcommand{\extractenvt}{\fnstyle{ExtractE}}
\newcommand{\extractcl}{\fnstyle{ExtractCl}}
\newcommand{\cachetoenvt}{\fnstyle{CacheToEnvt}}

\newcommand{\evalval}{\fnstyle{Eval}}
\newcommand{\tochc}{\fnstyle{toCHC}}
\newcommand{\pathcond}{\fnstyle{Cond}}
\newcommand{\map}{M}
\newcommand{\id}{\fnstyle{Id}}
\defgn[evals]{V}
\defgn[revals]{R}
\newcommand{\la}{\langle}
\newcommand{\ra}{\rangle}

\newcommand{\vdashc}{\vdash_\textrm{c}}
\newcommand{\vdashal}{\vdash_\textrm{d}}
\newcommand{\vdashcc}{\vdash_\textrm{cc}}

\newcommand{\vdashe}{\vdash_\textrm{e}}
\newcommand{\vdasha}{\vdash_\textrm{a}}
\newcommand{\vdashaa}{\vdash_\textrm{aa}}
\newcommand{\vdashea}{\vdash_\textrm{ea}}
\newcommand{\nofl}{\fnstyle{NoFL}}

\newcommand{\lC}{\listConcat}
\newcommand{\lCk}{\mathrel{\mbox{$\lC$}_k}}
\newcommand{\suffixes}{\fnstyle{Suffixes}}

\newcommand{\ruleref}[1]{\textsc{#1}}

\usepackage{listings}

\newcommand{\plangbasicstyle}{\ttfamily}
\newcommand{\plangkeywordstyle}{\plangbasicstyle \color{NavyBlue}}
\newcommand{\plangsymbolstyle}{\plangbasicstyle \color{RoyalPurple}}
\newcommand{\plangcommentstyle}{\plangbasicstyle \color{ForestGreen}}

\newcommand{\ttob}{\text{\upshape\ttfamily\char`\{}}
\newcommand{\ttcb}{\text{\upshape\ttfamily\char`\}}}
\newcommand{\ttop}{\text{\upshape\ttfamily(}}
\newcommand{\ttcp}{\text{\upshape\ttfamily)}}

\catcode`\#=11%
\newcommand{\es}{} 
\newcommand{\plangset}{%
    \lstset{%
        basicstyle=\plangbasicstyle,
        fontadjust=false,
        showspaces=false,
        showtabs=false,
        numberstyle=\tiny\color{gray},
        stepnumber=1,
        numbers=left,
        numbersep=5pt,
        escapeinside={$}{$},
        escapebegin=$,
        escapeend=$,
        keywordstyle=\plangkeywordstyle,
        keywords={fun,if},
        commentstyle=\plangcommentstyle,
        morecomment=[l]{#},
        literate=%
            {\{}{\begingroup    \plangsymbolstyle   \ttob       \endgroup}{1}%
            {\}}{\begingroup    \plangsymbolstyle   \ttcb       \endgroup}{1}%
            {(}{\begingroup     \plangsymbolstyle   \ttop       \endgroup}{1}%
            {)}{\begingroup     \plangsymbolstyle   \ttcp       \endgroup}{1}%
            {=}{\begingroup     \plangsymbolstyle   =\es        \endgroup}{1}%
            {?}{\begingroup     \plangsymbolstyle   ?\es        \endgroup}{1}%
            {:}{\begingroup     \plangsymbolstyle   :\es        \endgroup}{1}%
            {~}{\begingroup     \plangsymbolstyle   \tttilde    \endgroup}{1}%
            {->}{\begingroup    \plangsymbolstyle   ->\es       \endgroup}{2}%
            {@\{}{\begingroup   \plangcommentstyle  \ttob       \endgroup}{1}%
            {@\}}{\begingroup   \plangcommentstyle  \ttcb       \endgroup}{1}%
            {@:}{\begingroup    \plangcommentstyle  :\es        \endgroup}{1}%
            {@=}{\begingroup    \plangcommentstyle  =\es        \endgroup}{1}%
    }%
}
\catcode`\#=6%

\plangset 
\lstnewenvironment{plang}[1][]{\begingroup\renewcommand{\plangbasicstyle}{\footnotesize\ttfamily} #1}{\endgroup}
\newcommand{\plangil}[1][]{\plangset{}#1\lstinline[mathescape]}

\newcommand{\camlset}{\lstset{language=caml,numbers=left,numbersep=10pt,frame=single,framerule=0pt,basicstyle=\ttfamily,escapeinside={$}{$},escapebegin=$,escapeend=$,morekeywords={case,input}}}
\newcommand{\camlnonumset}{\lstset{language=caml,numbers=none,xleftmargin=10pt,frame=single,framerule=0pt,basicstyle=\ttfamily,escapeinside={$}{$},escapebegin=$,escapeend=$,morekeywords={case,input}}}

\lstnewenvironment{caml}{\begingroup \camlset}{\endgroup}
\lstnewenvironment{camlnonum}{\begingroup \camlnonumset}{\endgroup}
\newcommand{\camlil}[1][]{\camlset{}#1\lstinline[mathescape]}

\makeatletter
\lst@Key{basewidth}{0.5em,0.35em}{\lstKV@CSTwoArg{#1}
    {\def\lst@widthfixed{##1}\def\lst@widthflexible{##2}%
     \ifx\lst@widthflexible\@empty
         \let\lst@widthflexible\lst@widthfixed
     \fi
     \def\lst@temp{\PackageError{Listings}%
                                {Negative value(s) treated as zero}%
                                \@ehc}%
     \let\lst@error\@empty
     \ifdim \lst@widthfixed<\z@
         \let\lst@error\lst@temp \let\lst@widthfixed\z@
     \fi
     \ifdim \lst@widthflexible<\z@
         \let\lst@error\lst@temp \let\lst@widthflexible\z@
     \fi
     \lst@error}}
\makeatother

\newcommand{\codefigurestart}[1][.4\textwidth]{\hspace*{15pt}\begin{minipage}{#1-20pt}\camlset\footnotesize}
\newcommand{\codefigurestop}{\end{minipage}}

\usepackage{microtype} 
\definecolor{ForestGreen}{rgb}{.132,.545,.132}
\definecolor{Plum}{rgb}{.868,.628,.868}
\definecolor{RoyalPurple}{rgb}{.38,.25,.6}
\definecolor{NavyBlue}{rgb}{0,0,.5}
\definecolor{VioletRed}{rgb}{.816,.125,.565}

\usepackage{booktabs} 

\usepackage[ruled]{algorithm2e} 

\SetAlFnt{\small}
\SetAlCapFnt{\small}
\SetAlCapNameFnt{\small}
\SetAlCapHSkip{0pt}
\IncMargin{-\parindent}

\newcommand\eg{\emph{e.g.},\ }
\newcommand\ie{\emph{i.e.},\ }

\nonotes
\long\def\tossit#1{}
\long\def\keepit#1{#1}

\def\fullversion{\long\def\infull{\keepit}\long\def\inshort{\tossit}}
\fullversion
 
\setcopyright{rightsretained}
\acmDOI{10.1145/3649852}
\acmYear{2024}
\copyrightyear{2024}
\acmSubmissionID{oopslaa24main-p150-p}
\acmJournal{PACMPL}
\acmVolume{8}
\acmNumber{OOPSLA1}
\acmArticle{135}
\acmMonth{4}
\received{21-OCT-2023}
\received[accepted]{2024-02-24} 



\begin{document}

\title[A Pure Demand Semantics]{A Pure Demand Operational Semantics}\subtitle{with Applications to Program Analysis}
 
\author{Scott Smith} 
\orcid{0009-0005-0495-2716}
\affiliation{%
  \institution{Johns Hopkins University}
  \city{Baltimore}
  \country{USA}
}
\email{scott@cs.jhu.edu}

\author{Robert Zhang}
\orcid{0009-0001-8853-5813}
\affiliation{%
  \institution{Johns Hopkins University}
  \city{Baltimore}
  \country{USA}
}
\email{jzhan239@jhu.edu}




\begin{abstract}
    This paper develops a novel minimal-state operational semantics for higher-order functional languages that uses \emph{only} the call stack and a source program point or a lexical level as the \emph{complete} state information: there is no environment, no substitution, no continuation, etc.  We prove this form of operational semantics equivalent to standard presentations.
    
    We then show how this approach can open the door to potential new applications: we define a program analysis as a direct finitization of this operational semantics.  The program analysis that naturally emerges has a number of novel and interesting properties compared to standard program analyses for higher-order programs: for example, it can infer recurrences and does not need value widening.  We both give a formal definition of the analysis and describe our current implementation.
\end{abstract}

\maketitle

\section{Introduction}
\label{sec_intro}

There are numerous well-known approaches whereby operational semantics can be defined for functional programming languages: substitution-based reduction from the original $\lambda$-calculus \cite{LambdaBar}, substitution-based call-by-name or call-by-value reduction \cite{PlotkinSOS,Plotkincbncbv}, and various machine-based semantics including  SECD, Krivine, CESK \cite{FelleisenCEK}, and G-machines.  In this paper we define a minimal-state presentation of operational semantics that dispenses with substitutions, environments, stores, and continuations, and instead relies \emph{only on the call stack and a source program point or a lexical level} to encapsulate all information about the current state of a computation.

But how can variable values be looked up if the only dynamic state information is the call stack?  In fact, there is enough information in the call stack to trace back in time through the control flow to the original point where a variable was defined.  Because no values are propagated forward, we call it a \emph{pure demand semantics}: all values are looked up on demand. 
Note that pure demand is orthogonal to the laziness of lazy functional languages.  Laziness describes the time during evaluation at which an expression is computed.  On-demand is an orthogonal dimension describing how eagerly a variable-to-expression/value mapping is propagated.  Pure demand semantics can be defined for both call-by-name and call-by-value evaluation.

For an example of the minimality of a pure demand semantics, consider a factorial definition and invocation:

\begin{caml}
                  let rec fact n =
                    (if n = 0 then 1 else (fact (n - 1))$^{\reallab 2}$ * n)$^{\reallab 3}$
                  in (fact 4)$^{\reallab1}$
\end{caml}

The green superscripts $\reallab 1$, $\reallab 2$, and $\reallab 3$ abbreviate the enclosed expressions: $\reallab 1 = $ \camlil!fact 4! for example.  For this program run, an example state of execution in our operational semantics would be simply $[\reallab 2,\reallab 2,\reallab 1] \vdash \reallab 3 \steps \dots$ where $\stk = [\reallab 2,\reallab 2,\reallab 1]$ is the current call stack (a list of call sites in the source program) and $\reallab 3$ is the expression currently being evaluated, invariably a \emph{source} program expression.  That is \emph{all} the state information that is needed for non-variable lookup (for variable lookup, only a lexical level number is needed).  This call stack means the current stack state is an initial call at $\reallab 1$ followed by two recursive calls at $\reallab 2$.  Values for local variables such as \texttt{n} are then looked up by walking back up the stack. Looking up \texttt{n} in the body will be seen to be \texttt{(4 - 1) - 1} by chaining back through the arguments at the call sites at $\reallab 2, \reallab 2, \reallab 1$ in turn.

Accessing non-local variable values in pure demand semantics is more complex than local variable access.  The method we use bears some similarity to access links used in implementations of imperative languages with nested function definitions (but no fully first-class functions) \cite{Fischer05,MitchellCompilers}: non-local variables are accessed by chaining out to the lexical frame in which they were defined.   Some previous demand-driven presentations of operational semantics also dispense with environments and stores \cite{DDPA,DDSE}, but these systems still need additional state information beyond the call stack.  Pure demand semantics have an \emph{extremely} minimal state.

We prove that our pure demand approach is equivalent to a standard environment/closure big-step operational semantics: it is not changing the meaning of computation, only the means.

So, what good is it?  First, it is arguably interesting purely as mathematics: it is a maximal compression of program state information. But we believe it should also make possible potential new constructions; we will present one such construction in this paper.

\paragraph{A Pure Demand Program Analysis}  In the second part of this paper, we will present a specific application to give one detailed example to show how pure demand operational semantics is potentially useful: we define a pure demand program analysis, which is a direct finitization of the pure demand operational semantics \emph{a la} Abstract Interpreters for Free \cite{MightAbstractInterpretersForFree}.  Standard constructions of abstract interpreters for functional programming languages define a means to finitize a standard operational semantics, and so guarantee a finite environment/stack/store/etc. structure.  In our case we \emph{only} need to finitize the call stack, and for that we can use the standard notion of keeping only the $k$-most-recent frames.  We use a novel \emph{stack stitching} heuristic to approximate longer stacks by stitching together overlapping $k$-length fragments to gain more precision.

The analysis that emerges has some interesting and unusual properties.  There is no value widening, and as such there is no need to predefine any finite value domain as is usually the case for program analyses.  The analysis naturally infers constant values, sets of constants, intervals, and even recurrence relations on integer values.  


We will additionally show how a pure demand program analysis can be formally defined \emph{solely} by a proof system, without need for additional pseudocode required in other presentations \cite{AAM} to specify how paths are collected and how repeating states are pruned.  We call our proof system an \emph{all paths} system because it considers all nondeterministic execution path choices in parallel.

We describe our initial implementation of this analysis and show that it performs correctly on common benchmarks for functional program analyses.

\paragraph{Summary of Contributions}
Here we summarize the key contributions of this work.
\begin{enumerate}
    \item The definition of a novel ``pure demand'' form of operational semantics which propagates no variable bindings forward;
    \item A proof of its correctness relative to a standard propagation semantics;
    \item  The definition of a program analysis directly  from the pure demand semantics, which naturally infers general value abstractions including recurrences;
    \item The first known fully rule-based presentation of a higher-order program analysis;
    \item The stack-stitching technique used in the analysis, which gives a new approach to the venerable question of call stack finitization;
    \item An implementation that confirms the soundness and potential usefulness of the ideas.
\end{enumerate}

\paragraph{Outline}  Section \ref{sec_opsem} focuses on the operational semantics and Section \ref{sec_soundness} establishes the (non-trivial) proof of its equivalence to a standard environment-based operational semantics.  Section \ref{sec_analysis} defines the program analysis.  Section \ref{sec_implementation} describes the implementation.  Lastly, we cover related work in Section \ref{sec_related} and conclude in Section \ref{sec_concl}.  All proofs are found in Appendix \infull{\ref{sec_proofs}}\inshort{A of the supplementary material}, and the implementation details are outlined in Appendix \infull{\ref{sec_implement_appendix}}\inshort{B of the supplementary material}.

\section{A Pure Demand Operational Semantics}
\label{sec_opsem}

In this section we define the pure demand operational semantics, first for only the pure lambda-calculus and then with additional features.  We will prove the equivalence of the pure demand and a standard operational semantics in the following section.

We somewhat arbitrarily choose to investigate call-by-value only here to make clear that ``demand'' is not the same as ``lazy''; call-by-name proceeds similarly.

\subsection{The Core Language}
\label{sec_core}
\begin{wrapfigure}{r}{.4\textwidth}
  \begin{grammar}
            \grule[expressions]{\expr}{
              (\expr \ \expr) \gor \eval \gor \ev^\dbn
            }
            \grule[DeBruijn indices]{\dbn}{
                        0\gor 1\gor 2\gor\ldots
            }             
            \grule[variables]{\ev}{
              \textit{(identifiers)}
  }         \grule[variable uses]{\_}{\ev^\dbn    
}            
            \grule[call sites aka frames]{\frm}{
              (\expr \ \expr)
  }
            \grule[values]{\eval}{
            \gtfun\ \ev \gtarrow \expr
            }            
             \grule[results]{\reval}{
              (\gtfun\ \ev \gtarrow \expr)^{\stk} 
  }
            \grule[call stacks]{\stk}{
                        [\frm, \ldots]
            }
        \end{grammar}
    \caption{Language and interpreter grammar}
    \label{fig_Grammar}
\end{wrapfigure}

The language grammar and auxiliary grammar needed for the operational semantics appears in Figure \ref{fig_Grammar}. We will assume some fixed closed expression $\gexpr$.

For variable uses, we use the notation $\ev^\dbn$ that represents a combination of DeBruijn indices and the regular notation: $\ev$ is the variable, and $\dbn$ is a DeBruijn index.  The first set of rules we give here in fact \emph{only} uses the DeBruijn indices $\dbn$ and ignores $\ev$, but soundness is proved by a series of translations through systems without DeBruijn indices so it simplifies to use standard named and scoped variables in addition to the indices.  For expression $\gexpr$ to be well-formed, we require first that for each variable occurrence $\ev^\dbn$ in the program, $\dbn$ is constrained to be the number of enclosing non-$\ev$ $\gtfun$ definitions for this particular occurrence of $\ev$ (\ie it is a DeBruijn index but starting the count at 0 instead of DeBruijn's 1).  
For example, $\gtfun\ \texttt{x} \gtarrow \texttt{x}^0$ and $\gtfun\ \texttt{x} \gtarrow \gtfun\ \texttt{y} \gtarrow (\texttt{x}^1\ \texttt{y}^0)$  are well-formed programs while $\gtfun\ \texttt{x} \gtarrow \texttt{x}^1$ is not.  Secondly, we require that all variable definition occurrences $\gtfun\ \ev \gtarrow\dots$ are unique (any expression repeating a variable binding could be $\alpha$-renamed to satisfy this restriction).  Lastly, all $\gexpr$'s must be closed.

Stacks $\stk$ are lists of frames $\frm$ that are simply call sites in the source program: $\stk = [\frm_1, \ldots, \frm_n]$.  List append is notated $[\frm_1, \ldots,\frm_n] \lC [\frm_{n+1}, \ldots,\frm_{n+n'}] = [\frm_1, \ldots,\frm_{n+n'}]$ and we overload $\frm_0 \lC [\frm_1, \ldots,\frm_n]$ to mean $[\frm_0] \lC [\frm_1, \ldots,\frm_n]$.  We define $\frm \in [\frm_1, \ldots,\frm_n]$ iff $\frm = \frm_i$ for some $1\leq i\leq n$.

Value results in the operational semantics are of the form $\reval = (\gtfun\ \ev \gtarrow \expr)^{\stk}$ with the stack $\stk$ at the control flow point where the value was defined.  $\stk$ can be viewed as a replacement for a closure and is used for looking up non-local variable values, as we will describe below.

\subsubsection{The Rules}

The rules are given in Figure \ref{fig_steps}.  We use a big-step semantics with judgements of the form $\stk \vdash \expr \steps \reval$.  The stack $\stk$ is the call stack at which $\expr$ is evaluated.  Notice there is neither substitution nor any environment involved; variable lookups per the \ruleref{Var} rules are achieved by walking back along the control flow to find the original definition point, and the call stack alone suffices for this.

\begin{definition}
  \label{def_basic_opsem}
   Relation $\stk \vdash \expr \steps \reval$ is defined by the rules of Figure \ref{fig_steps}.  For well-formed top-level programs $\gexpr$, we define $\vdash \gexpr \steps \reval$ iff $[] \vdash \gexpr \steps \reval$.
\end{definition}

Note that there is an invariant in the rules that $\expr$ appearing anywhere in the rules means it is a sub-expression of the original $\gexpr$ (\ie there is no substitution or relabeling), \emph{except} for the special case of variables, which may have a DeBruijn index smaller than the index in the source program.  Note that the total required state information is formally a list of call sites and \emph{either} a call site program point, or a DeBruijn index $\dbn$.  The underlying variable $\ev$ of $\ev^\dbn$ is ignored in the rules as can be seen by a simple induction: for the base case it is ignored in \ruleref{Var Local}, and by induction it is ignored in \ruleref{Var Non-Local}. So, formally only the DeBruijn index is needed to produce a proof.

This theory is tighly bound to a big-step presentation unlike standard environment or substution approaches, which can be elegantly presented in either big-step or small-step form.  Our lookup of non-local variables involves a context switch into the context of the defining function, and while this could in principle be possible in a small-step presentation, it would be convoluted.  Small-step does have some advantages we lose with big-step.  For example, our big-step semantics equates divergence and stuck-ness (runtime type error) as equivalently having no constructible proof tree.  Stuck-ness is expressible in big-step but it requires explicit error bubbling rules and so we have left it out for simplicity.

\begin{figure}
  \begin{mathpar}
    \bbrule{Value}{\ 
   }{\stk \vdash \gtfun\ \ev \gtarrow \expr \steps (\gtfun\ \ev \gtarrow \expr)^{\stk}
   }
    \bbrule{{Var Local}}{
          \stk \vdash \expr_2 \steps \reval
      }{((\expr_1\ \expr_2) \lC \stk) \vdash \ev^0 \steps \reval
      }

      \bbrule{{Var Non-Local}}{
        \stk \vdash \expr_1 \steps (\gtfun\ \ev_1 \gtarrow \expr )^{\stk_1} \\
        \stk_1 \vdash \ev^\dbn \steps \reval
      }{((\expr_1\ \expr_2) \lC \stk) \vdash \ev^{\dbn + 1} \steps \reval 
      }

      \bbrule{Application}{
        \stk \vdash \expr_1 \steps (\gtfun\ \ev \gtarrow \expr)^{\stk_1}\\
        \stk \vdash \expr_2 \steps \reval_2\\
        ((\expr_1\ \expr_2) \lC \stk) \vdash \expr \steps \reval
      }{\stk \vdash (\expr_1\ \expr_2) \steps \reval
      }
  \end{mathpar}
  \caption{Pure demand operational semantics rules}
  \label{fig_steps}
\end{figure}

\subsubsection{On-Demand Variable Lookup}
\label{sec_on_demand_lookup}

The key feature of this semantics is how variable lookup proceeds.  For a local variable $\ev^0$, we know that the top (\ie leftmost) frame of the stack is the call to this function, and thus the parameter, $\expr_2$ in the \ruleref{Var Local} rule, when evaluated is the value of $\ev^0$.  So, this is exactly what the rule does.

The most interesting rule is the \ruleref{Var Non-Local} rule.  Here we observe that the variable occurrence $\ev^{\dbn+1}$ is not a local variable as the DeBruijn index is greater than 0.  To find the value of $\ev^{\dbn+1}$, the philosophy of this rule is to look up the current function (which is the value of $\expr_1$ as the current function is the most recent call) and from this result obtain a stack $\stk_1$, which is the stack in which this function was \emph{defined} (the \ruleref{Value} rule implicitly enforces this invariant). At the point this function was defined, we are one lexical level closer to the function in which $\ev$ is a local, so we then look up $\ev$ \emph{from this new context}, which we achieve by decrementing the DeBruijn index by one: $\ev^\dbn$.  Eventually such a chain up lexical levels will lead us to DeBruijn index 0, and the \ruleref{Var Local} rule may then be used to look up the variable's value.

\subsubsection{An Example}
Figure \ref{fig_steps_eg} contains an example proof tree for $$\gexpr = (((\gtfun\ x \gtarrow (\gtfun\ y \gtarrow x^{\dbrlab 1})^{\reallab y})^{\reallab x}\ 1)^{\reallab {x1}}\ 2)^{\reallab{e}}$$ computing to $1$. We use green superscripts such as $\reallab {y}$ here as abbreviations for the enclosed subexpression: so \eg $\reallab {y} = \gtfun\ y \gtarrow x^{\dbrlab 1}$. The numerical superscripts on variables $x^{\dbrlab 1}$ are the DeBruijn indices; since $\ev$ is defined one lexical level out it has an index of $1$ here. We also use integer constants not actually in this simple language to make the example concise.  Lastly, for simplicity we leave out the middle hypothesis of the \ruleref{Application} rule which evaluates the argument as per the call-by-value calling convention, and throws it away.

The left proof tree here evaluates $\reallab{x1}$, the application of the function $\reallab x$ to $1$.  The return value here is the function $\reallab y$.  Then the right proof tree evaluates the body of this function, $\ev^{\dbrlab {1}}$, and pushes the call site, $\reallab e$, onto the stack since we are now evaluating the body of the function of this call.  $\ev^{\dbrlab{1}}$ is a non-local variable and so the \ruleref{Var Non-Local} rule applies here; the top stack frame $\reallab e$ is $(\reallab{x1}\ 2)$ and so we first look up function $\reallab{x1}$, and obtain result $\reallab y^{[\reallab{x1}]}$, the key being we obtain the stack $[\reallab{x1}]$ which is the calling context where this function was defined.  So, we then take this calling context and make it the current context and look up $\ev$ from this point, but decrementing the DeBruijn index to give $\ev^{\dbrlab 0}$, indicating it is being viewed one level up in scope where the $\reallab y$ function was defined. $\ev^{\dbrlab 0}$ is now a local in this context and the \ruleref{Var Local} rule then applies, grabbing the $1$ argument from the top frame $\reallab{x1}=(\reallab x\ 1)$ on the stack, and returning that as the result.  This simple example captures the essence of non-local variable lookup; for lexically deeper non-local variable references, repeated \ruleref{Var Non-Local} rule invocations will decrement the index by one until 0 and \ruleref{Var Local} is reached.

\begin{figure}\footnotesize
$$
\frac{
  \cfrac{
   \cfrac{}{[] \vdash\reallab x \steps \reallab x^{[]}}\ \ \ 
    \cfrac{}{[\reallab {x1}] \vdash \reallab y \steps \reallab y^{[\reallab {x1}]}}
  }
  {[] \vdash \reallab {x1} \steps \reallab y^{[\reallab {x1}]}}\ \ \ 
  \cfrac{\cfrac{\cfrac{}{[] \vdash\reallab x \steps \reallab x^{[]}}\ \ \ 
  \cfrac{}{[\reallab {x1}] \vdash \reallab y \steps \reallab y^{[\reallab {x1}]}}}{[] \vdash \reallab {x1} \steps \reallab y^{[\reallab {x1}]}}\ \ 
  \cfrac{\cfrac{}{[] \vdash 1\steps 1}}
  {[\reallab {x1}] \vdash x^{\dbrlab {0}} \steps 1}}
  {[\reallab {e}] \vdash x^{\dbrlab 1} \steps 1}
}{[] \vdash \biggl(\biggl(\bigl(\gtfun\ x \gtarrow (\gtfun\ y \gtarrow x^{\dbrlab 1})^{\reallab y}\bigr)^{\reallab x}\ 1\biggr)^{\reallab {x1}}\ 2\biggr)^{\reallab{e}}  \steps 1}
$$
  \caption{An example derivation using the rules of figure \ref{fig_steps}}
  \label{fig_steps_eg}
\end{figure}




\subsection{From Access Links to Displays}
\label{sec_displays}
The incremental chaining up lexical levels found in the \ruleref{Var Non-Local} rule bears some similarity to the access link chains in classic compiler implementations \cite{MitchellCompilers,Fischer05}: both incrementally proceed out lexical levels one by one until the point of definition of a needed variable.  Classic access links are still propagating values forward (and, they only capture languages with nested function definitions and not full higher-order functions); our pure demand theory can be viewed as a more demand and more general form of access links.

There is also a classic optimization to access links where rather than chaining back lexical levels, a table is instead pushed forward which contains the whole chain; this more eager approach is called a \emph{display} \cite{MitchellCompilers,Fischer05}.  In our theory, it is in fact possible to perform an optimization similar to how displays optimize access links in compiler runtimes, by eagerly propagating a table for lexical variable lookup.  We now present a formal theory reflecting this approach.  This theory is not used in the remainder of the paper and is primarily here to make the historical connection.

Note that displays end up being much more complex in the presence of higher-order functions, since displays themselves may contain displays. This complication may explain why displays have apparently only been used in first-order programming languages with lexically-nested function definitions.

In this system we again use the DeBruijn index grammar for variable uses as in Section \ref{sec_opsem}.

\begin{wrapfigure}{r}{.58\textwidth}%
    \begin{grammar}
            \grule[expressions]{\expr}{
              (\expr \ \expr) \gor \eval \gor \ev^\dbn \gor \expr \binop \expr
              \gor \expr \gtquestion \expr \gtcolon \expr
              \gline
              \gor \expr \gtdot \lbl 
              \gor \ob \lbl \gteq \expr; \ldots \cb
              \gor \lbl \gtin \expr
            }
            \grule[DeBruijn indices]{\dbn}{
            0\gor 1\gor 2\gor\ldots
            }   
            \grule[variables]{\ev}{
                        \textit{(identifiers)}
            }  
            \grule[variable uses]{\_}{\ev^\dbn    
            }
            \grule[call sites aka frames]{\frm}{
            (\expr \ \expr)
            }          
            \grule[record labels]{\lbl}{
              \textit{(record labels)}
            }
            \grule[integers]{\eint}{
              0 \gor 1 \gor -1 \gor \dots
            }
            \grule[booleans]{\ebool}{
              \gttrue \gor \gtfalse
            }            
            \grule[values]{\eval}{
              \gtfun\ \ev \gtarrow \expr
              \gor \ob \lbl \gteq \eval; \ldots \cb
              \gor \eint
              \gor \ebool
            }
            \grule[operators]{\binop}{
              \gtplus \gor \gtminus \gor \gtand \gor \gtor \gor \gtxor \gor \gtleq \gor \gtless \gor \gteq
            }
            \grule[result values]{\reval}{
              (\gtfun\ \ev \gtarrow \expr)^{\stk}
              \gor \ob \lbl \gteq \reval; \ldots \cb
              \gor \eint
              \gor \ebool
              \gline
              \gor (\reval \binop \reval)
              \gor \reval \gtdot \lbl 
              \gor \lbl \gtin \reval
            }
            \grule[call stacks]{\stk}{
              [\frm, \ldots]
  }
        \end{grammar}
    \caption{Extended language grammar}
    \label{fig_grammar_extended}
\end{wrapfigure}

We let $\display= [F_0,\ldots,F_n]$ be displays, which are lists of stack frames $F$.
Stack frames are now $F = (\expr_1\ \expr_2)^\display$, each call needs to be able to restore its display when we pop out of the function body.    For example, the leftmost frame $F_0$ of $[F_0,\ldots,F_n]$ is the local context and contains the display used to look up $\ev^0$.  The \ruleref{Var} rule then has immediate access to all of the non-local contexts in $\display$, and can simply select the $F_\dbn$ at the level $\dbn$ of the variable being looked up. 

\begin{figure}
  \begin{mathpar}
    \bbrule{Value}{\ 
   }{\display \vdashal (\gtfun\ \ev \gtarrow \expr) \steps(\gtfun\ \ev \gtarrow \expr)^{\display}
   }

    \bbrule{{Var}}{
      \display = [F_0, \ldots, F_\dbn,\ldots,F_n]\\
      F_\dbn = (\expr_\dbn\ \expr'_\dbn)^{\display_\dbn}\\
      \display_\dbn \vdashal \expr'_\dbn \steps \reval
    }{\display \vdashal \ev^\dbn\steps \reval 
    }

      \bbrule{Application}{
        \display \vdashal \expr_1 \steps (\gtfun\ \ev \gtarrow \expr)^{\display_1}\\
        \display \vdashal \expr_2 \steps \reval_2\\
        [(\expr_1\ \expr_2)^{\display}]\lC \display_1 \vdashal \expr \steps \reval
      }{\display \vdashal (\expr_1\ \expr_2) \steps \reval
      }
  \end{mathpar}
  \caption{Operational semantics with displays}
  \label{fig_steps_access_links}
\end{figure}

\begin{definition}
  $\display \vdashal \expr \steps \reval$, the displays version of $\stk \vdash \expr \steps \reval$, is defined by the rules of Figure \ref{fig_steps_access_links}.
\end{definition}

This system is not as demand-driven as $\vdash$ but it still preserves the property of not propagating forward any actual program values.

\subsection{An Extended Language}
\label{sec_extended_language}

In this section, we show how the pure demand approach easily scales to full functional languages.  Concretely, we flesh out the language of the previous section by adding records, integers, booleans and conditionals.  Language features without variable bindings are relatively easy to add to this theory as they largely mirror standard big-step operational semantics rules.  There is a design choice of how ``demand'' the operators should be, either demand in that \eg \texttt{3 + 2} returns \texttt{3 + 2}, or non-demand in that it returns \texttt{5}. Here we will somewhat arbitrarily present the demand version as we will need to use this form in a program analysis built on this system in Section \ref{sec_analysis}.  We define a separate function $\evalval$ with \eg $\evalval(\texttt{3 + 2})=\texttt{5}$, which pulls on these uncompleted result computations when required (such as when a boolean result is needed for a conditional branch). 

The extended grammar appears in Figure \ref{fig_grammar_extended}.  The source language of expressions $\expr$ is fairly standard but note we also include $\lbl \gtin \expr$ syntax for runtime checking of the presence of a field in a record; with this operation it is possible to easily encode variants as records.  For example $\ob \texttt{hd} \gteq \expr; \texttt{tl} \gteq \expr \cb$ plus $\ob \texttt{nil} \gteq 0 \cb$ can encode lists with $\texttt{nil} \gtin \expr$ encoding the check for $\expr$ being an empty list.

Result values here are extended to include record values, integers and booleans, and with the demand operator approach we are taking, the operators will not (yet) be evaluated.  The only cases where the concrete value of an operation is needed for execution to continue is the boolean value of a conditional guard, to know which conditional branch to take; or, for a function result, which could be inside a record and so will need record projections.  We will define a special function $\evalval$ that reduces the operators in value results $\reval$ to produce actual values $\eval$ when needed.

\begin{definition}[Result Evaluation]
  $\evalval$ is defined by induction via the clauses in Figure \ref{fig_evalval}.  The function is undefined for cases of type mismatch, \eg projecting a number or adding records.
\end{definition}

We may now define the operational semantics for the extended language.

\begin{definition}
  \label{def_extended_opsem}
  Relation $\stk \vdash \expr \steps \reval$ for the extended language is defined by the rules of Figure \ref{fig_steps_extended}.  For top-level programs $\gexpr$ we define $\vdash \gexpr \steps \reval$ iff $[] \vdash \gexpr \steps \reval$.
\end{definition}

\begin{figure}
  \begin{mathpar}
    \bbrule{Value}{\eval \text{ is not a function}
    }{\stk \vdash \eval \steps \eval
    }

    \bbrule{Value Fun}{
   }{\stk \vdash \gtfun\ \ev \gtarrow \expr \steps (\gtfun\ \ev \gtarrow \expr)^{\stk}
   }

    \bbrule{{Var Local}}{
          \stk \vdash \expr_2 \steps \reval
      }{((\expr_1\ \expr_2) \lC \stk) \vdash \ev^0 \steps \reval
      }

      \bbrule{{Var Non-Local}}{
        \stk \vdash \expr_1 \steps \reval_f \\
        \evalval(\reval_f)=(\gtfun\ \ev_1 \gtarrow \expr )^{\stk_1}\\
        \stk_1 \vdash \ev^{\dbn} \steps \reval
      }{((\expr_1\ \expr_2) \lC \stk)  \vdash \ev^{\dbn + 1} \steps \reval 
      }

      \bbrule{Operation}{
        \stk \vdash \expr_1 \steps \reval_1\\
        \stk \vdash \expr_2 \steps \reval_2
      }{\stk \vdash \expr_1 \binop \expr_2 \steps \reval_1 \binop \reval_2
      }

      \bbrule{Record Value}{
        \stk \vdash \expr_1 \steps \reval_1 \ldots \stk \vdash \expr_n \steps \reval_n
      }{\stk \vdash \ob \lbl_1 \gteq \expr_1; \dots; \lbl_n \gteq \expr_n \cb \steps \ob \lbl_1 \gteq \reval_1; \dots; \lbl_n \gteq \reval_n \cb
      }

      \bbrule{Record Project}{
        \stk \vdash \expr \steps \ob \lbl_1 \gteq \reval_1; \dots; \lbl_n \gteq \reval_n \cb\\
        \lbl = \lbl_i
      }{\stk \vdash \expr\gtdot \lbl \steps \ob \lbl_1 \gteq \reval_1; \dots; \lbl_n \gteq \reval_n \cb\gtdot \lbl
      }

      \bbrule{Record Inspect}{
        \stk \vdash \expr \steps \ob \lbl_1 \gteq \reval_1; \dots; \lbl_n \gteq \reval_n \cb
      }{\stk \vdash \lbl \gtin \expr \steps \lbl \gtin \ob \lbl_1 \gteq \reval_1; \dots; \lbl_n \gteq \reval_n \cb
      }

      \bbrule{Conditional}{
        \stk \vdash \expr \steps \reval\\
        \evalval(\reval) = \ebool\\
        \stk \vdash \expr_{\ebool} \steps \reval'
      }{\stk \vdash (\expr \gtquestion \expr_{\gttrue} \gtcolon \expr_{\gtfalse}) \steps \reval'
      }

      \bbrule{Application}{
        \stk \vdash \expr_1 \steps \reval_f\\
        \evalval(\reval_f)=(\gtfun\ \ev \gtarrow \expr)^{\stk'}\\
        \stk \vdash \expr_2 \steps \reval_2\\
        ((\expr_1\ \expr_2) \lC \stk) \vdash \expr \steps \reval
      }{\stk \vdash (\expr_1\ \expr_2) \steps \reval
      }
  \end{mathpar}
  \caption{Operational semantics rules for the extended language}
  \label{fig_steps_extended}
\end{figure}

Inspecting the rules, for the new syntax they are similar to standard big-step rules modulo the on-demand results returned; the function application and variable lookup logic is basically identical to the core language.

\subsubsection{Efficiency}
\label{sec_interp_efficiency}

\begin{wrapfigure}{r}{.4\textwidth}
  $\evalval(n)  =  n$\\
  $\evalval(\ebool)  =  \ebool$\\
  $\evalval((\gtfun\ \ev \gtarrow \expr)^{\stk})  =  (\gtfun\ \ev \gtarrow \expr)$\\
  $\evalval(\reval_1 \binop \reval_2) = \evalval(\reval_1) \binop \evalval(\reval_2)$\\
  $\evalval(\reval \gtdot \lbl_i) = \eval_i \text{ where } $\\
  \hspace*{3em} $\evalval(\reval) = \ob  \dots; \lbl_i \gteq \eval_i; \dots \cb$\\  
  $\evalval(\lbl \gtin \reval) = \gttrue \text{ if } $\\
  \hspace*{3em} $\evalval(\reval) = \ob \dots; \lbl \gteq \eval; \dots \cb $\\
  $\evalval(\lbl \gtin \reval) = \gtfalse \text{ if } \lbl \not \in \evalval(\reval)$\\
  $\evalval(\ob \lbl_1 \gteq \reval_1; \dots; \lbl_n \gteq \reval_n \cb) =$\\
  \hspace*{2em}  $\ob \lbl_1 \gteq \evalval(\reval_1); \dots; \lbl_n \gteq \evalval(\reval_n) \cb $
\caption{The definition of $\evalval(\reval)$}\label{fig_evalval}
\end{wrapfigure}

While we believe this form of interpreter is primarily of interest for its metamathematical properties and not as an actual implementation basis, it is worthwhile to ask how much runtime efficiency has been lost if an interpreter was implemented with this mechanism compared to standard methods.  If the evaluation relation is cached via a map $(\stk,\expr) \mapsto \reval$, and $\evalval$ is also cached as $\reval \mapsto \eval$, there will in fact be no duplicate computations.  The keys of the cache include the stack $\stk$, which is an unbounded list, but hash-consing should allow ``O(1)'' cache lookup.  So, with caching the speed should be on the same order of runtime complexity as a standard interpreter.  Space utilization is another matter however; while it is possible to apply garbage collection heuristics to the otherwise ever-growing cache, there is still a pathological worst case where a deep tail-call function that returns a function as a result will have the whole deep stack annotating the function result.  For example, if a function was returned from a tail call a million calls deep, the stack in the closure of the function result would have one million frames.

We have implemented a pure demand interpreter for this language; the implementation is briefly covered in Section \ref{sec_interp_impl}. 

\section{Soundness}
\label{sec_soundness}
In this section, we establish that the pure demand operational semantics is provably equivalent to a standard big-step environment/closure operational semantics.  The proof is non-trivial and shows there is a large conceptual gap between the two equivalent forms of operational semantics.  We will prove soundness only over the core language of Section \ref{sec_core}, because that is where nearly all the novelty is.

\begin{wrapfigure}{r}{.4\textwidth}%
    \begin{grammar}
            \grule[expressions]{\expr}{
              (\expr \ \expr) \gor \eval \gor \ev
            }
            \grule[variables]{\ev}{
                        \textit{(identifiers)}
            }            
            \grule[values]{\eval}{
            \gtfun\ \ev \gtarrow \expr 
            }            
             \grule[results aka closures]{\reval}{
            \eval^{\envt} 
  }
            \grule[environments]{\envt}{
                        [\ev \mapsto \reval,\dots]
            }
        \end{grammar}
    \caption{Environment semantics grammar}
    \label{fig_envt_grammar}
\end{wrapfigure}

We first will define a standard environment/closure-based operational semantics $\vdashe$, the goal being to show that this standard system is equivalent to the core system $\vdash$ of Definition \ref{def_basic_opsem}.  The two systems differ considerably so we will define two intermediate systems $\vdashc$ and $\vdashcc$, and then show $\mathnormal\vdash \Leftrightarrow \mathnormal\vdashe$ by establishing three lemmas: $\mathnormal\vdash \Leftrightarrow \mathnormal\vdashc$, $\mathnormal\vdashc \Leftrightarrow \mathnormal\vdashcc$, and $\mathnormal\vdashcc \Leftrightarrow \mathnormal\vdashe$.

Complete proofs of all lemmas are given in Appendix \infull{\ref{sec_proofs}}\inshort{A, which is part of the supplemental information}.

\subsection{Environment/Closure-Based System}
\label{sec-envt_system}

We first start off with $\vdashe$, a standard call-by-value environment/closure big step operational semantics.

We let $\envt = [\ev_0 \mapsto \reval_0,\dots,\ev_n \mapsto\reval_n]$ represent an environment mapping variables to their values.  Since the language has no \plangil!let! and only functions, we can without loss of generality assume that the head variable mapped, $\ev_0$, is the (single) local variable, $\ev_1$ is the (single) non-local variable defined in the immediately containing static scope, etc.  Closures are results of the form $\reval = (\gtfun\ \ev \gtarrow \expr)^{\envt}$ where the $\envt$ is of the aforementioned form.  We no longer need the DeBruijn indices on variables so they are removed. The revised grammar is in Figure \ref{fig_envt_grammar}.  Note that we are reusing metavariable names such as $\expr$ across the different systems but it should be clear from context whether we are in $\vdash$, $\vdashe$, etc.

The big-step operational semantics is completely standard. 
\begin{definition}
    Relation $\envt \vdashe \expr \steps \reval$ is defined by the rules of Figure \ref{fig_steps_envt_closure}. $\vdashe \gexpr \steps \reval$ iff $[] \vdashe \gexpr \steps \reval$.
\end{definition}

\begin{figure}
  \begin{mathpar}
    \bbrule{Value}{\ 
   }{\envt \vdashe (\gtfun\ \ev \gtarrow \expr) \steps (\gtfun\ \ev \gtarrow \expr)^{\envt}
   }

   \bbrule{{Var}}{
\  }{[\ev_0 \mapsto \reval_0,\dots,\ev_n \mapsto\reval_n]  \vdashe \ev_i \steps \reval_i
 }


      \bbrule{Application}{
        \envt  \vdashe \expr_1 \steps (\gtfun\ \ev \gtarrow \expr)^{\envt_1}\\
        \envt  \vdashe \expr_2 \steps \reval_2\\
       [\ev \mapsto \reval_2] \lC \envt_1  \vdashe \expr \steps \reval
      }{\envt  \vdashe (\expr_1\ \expr_2) \steps \reval
      }
    \end{mathpar}
  \caption{Operational semantics with environment/closure}
  \label{fig_steps_envt_closure}
\end{figure}

\subsection{Chaining Variable Lookups Directly}

We will now proceed to define two intermediate systems $\vdashc$ and $\vdashcc$ to link $\vdash$ with $\vdashe$.  In this section, we define $\vdashc$ which will make the variable lookup of $\vdash$ a step closer to a standard operational semantics by making a single \ruleref{Var} rule that does all the scope chaining in one rule.

\begin{definition}
  $\stk \vdashc \expr \steps \reval$, the chaining version of $\stk \vdash \expr \steps \reval$, is defined by the rules of Figure \ref{fig_steps_chaining}.
\end{definition}

\begin{figure}
  \begin{mathpar}
    \bbrule{Value}{\ 
   }{\stk \vdashc (\gtfun\ \ev \gtarrow \expr) \steps (\gtfun\ \ev \gtarrow \expr)^{\stk}
   }

    \bbrule{{Var}}{
      n \geq 0\\
      \forall 0 \leq i < n.\ \stk_i \vdashc \expr_i \steps (\gtfun\ \ev_{i} \gtarrow \expr_{i}'' )^{\stk_{i+1}'} \\
      \forall 0 \leq i \leq n.\ \stk_{i}' = ((\expr_{i}\ \expr'_{i}) \lC \stk_{i})\\
      \stk_n \vdashc \expr'_n \steps \reval
    }{\stk_0' \vdashc \ev_n \steps \reval 
    }

      \bbrule{Application}{
        \stk \vdashc \expr_1 \steps (\gtfun\ \ev \gtarrow \expr)^{\stk_1}\\
        \stk \vdashc \expr_2 \steps \reval_2\\
        ((\expr_1\ \expr_2) \lC \stk) \vdashc \expr \steps \reval
      }{\stk \vdashc (\expr_1\ \expr_2) \steps \reval
      }
  \end{mathpar}
  \caption{Operational semantics with chaining variable lookup}
  \label{fig_steps_chaining}
\end{figure}

The grammar of this system is the same as the original in Figure \ref{fig_Grammar}, except the DeBruijn indexes $\dbn$ are no longer needed anywhere so are removed.  This is possible because we are chaining all of the \ruleref{Var Non-Local} rules together directly and the indices were only required to make this connection.

It is clear in any proof in $\vdash$ that any variable lookup starting at the conclusion node must have $n \geq 0$ \ruleref{Var Non-Local} rules in a row followed by a single \ruleref{Var Local} rule.  This holds because the only rules in $\vdash$ for looking up variable values are these two rules, and each use of \ruleref{Var Non-Local} will induce a lookup of the variable at one smaller index and so must eventually end in a \ruleref{Var Local}.

The intuition of the new \ruleref{Var} rule is that it is a ``proof tree macro'' expressing the chain described above in a single rule.  The DeBruijn index $\dbn$ reflects how deeply the variable $\ev$ is lexically nested and is aligned with the $n$ in the rule (but counting down and not up). The $n=0$ case in the rule corresponds to a local variable.  In this case there are no function definitions to chain through, and we can immediately look up $\ev_0$ by looking up $\expr_0'$.  For the case where $n=1$, it means $\ev_1$ is occurring immediately inside a $\gtfun\ \ev_0 \gtarrow\dots$, which occurs immediately inside a $\gtfun\ \ev_1 \gtarrow\dots$.  In this case, the lookup proceeds in the same spirit as the original system: we first look up the function $\expr_0$, which as result has stack $\stk_1'$, then we take the top frame on this stack, $(\expr_1\ \expr_1')$, and $\expr_1$ in turn must evaluate to a function with $\ev_1$ as parameter by the first $\forall$ condition and so the value of $\ev_1$ is then the value of $\expr_1'$, $\reval$.  In general, this pattern repeats for arbitrary levels of lexical nesting.

Before proving the systems equivalent, we characterize the canonical form that variable lookup subtrees must take in $\vdash$.  Figure \ref{fig_variable_lookup_canonical} expresses this canonical form.  In this figure, all rule applications are \ruleref{Var Non-Local} except for the upper right rule, which is \ruleref{Var Local}.  Recall that the proof of this lemma and other assertions below are found in \inshort{the supplementary material}\infull{Appendix \ref{sec_proofs}}; the above intuition constitutes a high-level sketch of the proof.

\begin{restatable}{lemma}{variablelookupcanonical}
  \label{lem_variable_lookup_canonical}
  Any variable lookup subtree in a proof of $\stk \vdash \expr \steps \reval$ must be of the form of Figure \ref{fig_variable_lookup_canonical}.  
\end{restatable}

\begin{figure}
  \resizebox{4.8in}{!}{
  \begin{mathpar}
    \frac{
      \stk_0 \vdash \expr_0 \steps (\gtfun\ \ev_0 \gtarrow \expr_0'' )^{(\expr_1 \ \expr_1')\lC\stk_1} \ \ \ 
      \cfrac{\stk_1\vdash \expr_1 \steps (\gtfun\ \ev_1 \gtarrow \expr_1'' )^{(\expr_2 \ \expr_2')\lC\stk_2}\ \ \ 
      \cfrac{\cfrac{\cfrac{\stk_n \vdash \expr_n' \steps \reval}{((\expr_n\ \expr_n') \lC \stk_n) \vdash \ev^0 \steps \reval}}{\stk_2\vdash \expr_2 \steps \ldots \ \ \vdots}}{(\expr_2\ \expr_2')\lC \stk_2 \vdash \ev^{\dbn-2} \steps \reval}
      }{(\expr_1\ \expr_1')\lC \stk_1 \vdash \ev^{\dbn-1} \steps \reval}
    }{((\expr_0\ \expr_0') \lC \stk_0) \vdash \ev^{\dbn} \steps \reval 
    }
  \end{mathpar}
  }
  \caption{Canonical variable lookup proof tree structure of $\vdash$}
  \label{fig_variable_lookup_canonical}
\end{figure}

We can now prove the two systems equivalent.
\begin{restatable}{lemma}{lemvc}
  \label{lem_v_c}
$\stk \vdash \expr \steps \reval$ iff $\stk \vdashc \expr \steps \reval$.
\end{restatable}

\subsection{Caching Lookups}

For the next step, we are going to define a set of rules that caches all of the function and argument lookups performed in the \ruleref{Var} rule of the previous system, producing the $\vdashcc$ system.  With this change, the proof tree should be structurally isomorphic to a standard environment/closure semantics proof tree for the same initial expression.  Note this is not a global cache, we only cache lookups of lexically-enclosing functions (in $\cachef$) and arguments (in $\cachel$); these lookups are the ones ``cached'' in an environment/closure semantics and thus are the ones we need to get proof alignment.

Concretely, we add a cache $\cachef = [\la \stk_0,\cachef_0,\cachel_0,\expr_0\ra \mapsto \reval_0, \dots]$ to each proof node to cache function lookups, and a cache $\cachel = [\la \stk,\cachef,\cachel,\expr\ra \mapsto (\reval,\ev)]$ to cache the local argument lookup, and the judgement form in this new cached chaining system will be $\stk, \cachef, \cachel \vdashcc \expr \steps \reval$.  In order to get a strong enough induction invariant for direct proofs, the caches also need to cache previous cache states; in the above form, it can be seen that we are precisely capturing other judgements in the cache.  In particular, we will need an invariant that lookups in a cache are all $\vdashcc$-provable \emph{in the cache that we saved in the cache}: for each $\la \stk_0,\cachef_0,\cachel_0,\expr_0\ra \mapsto \reval_0 \in \cachef$, we have $\stk_0,\cachef_0,\cachel_0 \vdashcc \expr_0 \steps \reval_0$.

We do not need to cache non-local arguments because we will include these local argument caches in our function closures, so $\cachel$ will either be a singleton list or an empty list (the latter for when we are at the top of the program).  $\cachel$ also needs to remember the function parameter $\ev$ along with the result $\reval$.  We additionally need to add $\cachef$ and $\cachel$ to the closures so they are now of the form $(\gtfun\ \ev \gtarrow \expr)^{\stk,\cachef,\cachel}$.  The intuition is the function cache will cache lookups $\stk_i \vdashc \expr_i \steps \reval_i$ so they don't need to be repeated in the \ruleref{Var} rule, and similarly the local argument lookup $\stk \vdashc \expr_0' \steps \reval$ will be cached.  The function cache we represent as a list and not a set of mappings because the rules invariably put the cached items with the innermost static scope at the front, next-innermost next, etc.

\begin{definition}$\stk, \cachef, \cachel \vdashcc \expr \steps \reval$ is defined by the rules of Figure \ref{fig_steps_caching}.
\end{definition}
Observe that only the \ruleref{Application} rule adds items to the cache mappings, and it adds the lookup of the function and argument that we just performed.  So, it is indeed just a cache recording previously-proved things.
 We write $\lC_{0\leq i \leq m} [X_i]$ as a shorthand for $[X_0,\ldots, X_m]$.

\begin{figure}
  \begin{mathpar}
    \bbrule{Value}{\ 
   }{\stk, \cachef, \cachel \vdashcc (\gtfun\ \ev \gtarrow \expr) \steps (\gtfun\ \ev \gtarrow \expr)^{\stk,\cachef,\cachel}
   }

    \bbrule{{Var}}{
      0 \leq n \leq m\\
      \cachef = \lC_{0 \leq i \leq m}[\la \stk_i,\cachef_i',\cachel_i',\expr_i\ra \mapsto (\gtfun\ \ev_{i} \gtarrow \expr_{i}'' )^{\stk_{i+1}',\cachef_{i+1},\cachel_{i+1}}] \\
      \forall 0 \leq i < n.\ \stk_{i+1}' = ((\expr_{i+1}\ \expr'_{i+1}) \lC \stk_{i+1})\\
      \cachel_n = [\la \stk_n,\cachef_n',\cachel_n',\expr_n'\ra \mapsto (\reval,\ev_n)]\\
    }{((\expr_0\ \expr_0')\lC \stk_0), \cachef, \cachel_0 \vdashcc \ev_n \steps \reval 
    }

      \bbrule{Application}{
        \stk, \cachef, \cachel \vdashcc \expr_1 \steps (\gtfun\ \ev \gtarrow \expr)^{\stk_1,\cachef',\cachel'}\\
        \stk, \cachef, \cachel \vdashcc \expr_2 \steps \reval_2\\
        ((\expr_1\ \expr_2) \lC \stk), (\la \stk,\cachef,\cachel,\expr_1\ra \mapsto (\gtfun\ \ev \gtarrow \expr)^{\stk_1,\cachef',\cachel'} \lC\cachef'), [\la \stk,\cachef,\cachel,\expr_2\ra \mapsto (\reval_2,\ev)] \vdashcc \expr \steps \reval
      }{\stk, \cachef, \cachel \vdashcc (\expr_1\ \expr_2) \steps \reval
      }
  \end{mathpar}
  \caption{Operational semantics with chaining and caching}
  \label{fig_steps_caching}
\end{figure}

This system is nearly identical to the previous one but rather than repeating the exact same function lookups deeper in the proof we simply cache them, and look up in the cache rather than repeat the sub-proof. 

All of the proof systems are deterministic; we in particular need the fact that $\vdashc$ is deterministic in the lemma below.  The proof is direct by induction.

\begin{lemma}
  \label{lem_deterministic}
$\vdash$, $\vdashc$, $\vdashcc$, and $\vdashe$ are all deterministic: given $\expr$ and all the assumptions to the left of the turnstile, there is at most one proof tree that can be constructed.
\end{lemma}

We can now prove $\vdashc$ and $\vdashcc$ equivalent.  This proof requires a strengthened induction invariant concerning properties of the cache.  The intuition of correctness is largely that of the above description: $\vdashcc$ is a formalization of caching for $\vdashc$.

\begin{restatable}{lemma}{lemccc}
  \label{lem_c-cc}
  For all $\expr$, $[] \vdashc \expr \steps(\gtfun\ \ev \gtarrow \expr)^{\stk}$ iff $[], [], [] \vdashcc \expr \steps (\gtfun\ \ev \gtarrow \expr)^{\stk,\cachef,\cachel}$.
\end{restatable}

\subsection{Relating the Cached and Environment Systems}

The $\vdashcc$ system is very close to the standard environment/closure system of Section \ref{sec-envt_system}, because the caches indirectly have all of the environment mappings cached in them.  In particular, we can define a function $\extractenvt$ which will extract an environment of variable-result bindings given the current stack and caches of $\vdashcc$.  The definition is inductive because function values themselves need to have their closure environments extracted.  The definition is well-founded based on lexical nesting depth.

\begin{definition}
Inductively define $\extractenvt$ and $\extractcl$ as follows.

Let $\extractenvt(\stk_0,\cachef,\cachel_0)$ be 
$[\ev_0 \mapsto \eval_0^{\extractcl(\reval_0)}, \dots, \ev_n \mapsto \eval_n^{\extractcl(\reval_n)}]$
where 
$$\cachef = [\la \stk_0,\expr_0\ra \mapsto (\gtfun\ \ev_{0} \gtarrow \expr_{0}'' )^{\stk_{1}',\cachef_{1}'\cachel_{1}}, \dots, \la \stk_m,\expr_m\ra \mapsto (\gtfun\ \ev_{m} \gtarrow \expr_{m}'' )^{\stk_{m+1}',\cachef_{m+1}',\cachel_{m+1}} ]$$
and where for all  $0\leq i < m$, $\stk_{i+1}' = (\expr_{i+1}\ \expr'_{i+1}) \lC \stk_{i+1}$, $\cachel_i = \la \stk_i,\expr_i'\ra \mapsto (\reval_i,\ev_i)$, and $\reval_i = \eval_i^{\stk''_i,\cachef_i'',\cachel_i''}$.
Let $\extractcl((\gtfun\ \ev \gtarrow \expr)^{[],[],[]})$  be $[]$ and let $\extractcl((\gtfun\ \ev \gtarrow \expr)^{(\expr_0 \expr_0')\lC\stk,\cachef,\cachel})$ 
be $\extractenvt(\stk,\cachef,\cachel)$. 
\end{definition}
  
The $\extractenvt$ function links the $\vdashe$ environment system with the $\vdashcc$ cache chaining system.  Concretely, we have the following lemma. Again, the intuition for the isomorphism is above: all the environments in the $\vdashe$ proof are in fact nested in the caches of the $\vdashcc$ proof.  The proof itself requires a strengthening of the induction hypothesis.

\begin{restatable}{lemma}{lemcce}
  \label{lem_cc-e}
  Given fixed source program $\gexpr$, $[],[] \vdashcc \gexpr \steps (\gtfun\ \ev \gtarrow \expr)^{\stk,\cachef,\cachel}$ iff $[] \vdashe \gexpr \steps (\gtfun\ \ev \gtarrow \expr)^{\envt}$, where $\extractenvt(\stk,\cachef,\cachel)=\envt$.
\end{restatable}

Putting together the lemmas relating these different systems we directly have the following.

\begin{restatable}{theorem}{thmve}
  \label{thm_v_e}
Given fixed source program $\gexpr$, $[] \vdash \gexpr \steps (\gtfun\ \ev \gtarrow \expr)^{\stk}$ iff $[] \vdashe \gexpr \steps (\gtfun\ \ev \gtarrow \expr)^{\envt}$.
\end{restatable}

\section{A Pure Demand Program Analysis}
\label{sec_analysis}
We now develop an application of pure demand operational semantics: it can be used to derive a program analysis by \emph{only} having to finitize the possible call stack states since that is the only unbounded state in a judgement; in a standard program analysis, the environment and store also need to be finite.  

We will force the stacks to the top-$k$ frames for a fixed $k$, a simple and standard stack finitization found in many program analyses \cite{ShiversThesis,MightAbstractInterpretersForFree,AAM,P4F}.  In these standard analyses the call stack is only needed to check for repeated states and record the return point, but in our system the stack is critical to demand-driven variable lookup.  For this reason, just keeping the $k$ most recent stack frames for us would lead to a great deal of inaccuracy in variable lookup: with a deep enough call stack there would be a complete loss of information on where a variable value originated from.  But, we can recover more precision by remembering all stack fragments visited thus far, and ``stitching'' existing stack fragments together where one fragment's suffix matches another fragment's prefix.

As with the operational semantics, we will first study the functions-only version and prove properties of that system.  We will then extend it to the larger language syntax of Section \ref{sec_extended_language}.

\subsection{The Core Analysis}
\label{sec_core_analysis}

We will use the same grammar as for the core language of Figure \ref{fig_Grammar}.  We use notation $\lC_k$ to denote a list consing that is then shortened to at most $k$ elements by dropping elements from the rear of the list: $\frm_0 \lC_k [\frm_1,\ldots,\frm_{k-1},\ldots] = [\frm_0,\frm_1,\ldots,\frm_{k-1}]$ and $\frm_0 \lC_k [\frm_1,\ldots,\frm_m] = [\frm_0,\frm_1,\ldots,\frm_{m}]$ where $m < k$.  We will use this function to drop some stack frames $\frm$ at calls in order to finitize the analysis. It must be the case that  $k \geq 1$, and practically speaking $k \geq 2$ is needed or no stitching will occur.  We must keep a history of all maximally-$k$-length stack fragments $\stk$ seen so far in the derivation in set $\sfrags=\{\stk_1,\ldots,\stk_n\}$, while maintaining an invariant that any actual runtime stack has to be a concatenation of fragments from $\sfrags$.

\begin{definition}
\label{def_core_analysis}
Given fixed $k>0$, relation $\stk, \sfrags \vdasha \expr \steps \reval$ is defined by the rules of Figure \ref{fig_steps_analysis}.  For top-level programs $\gexpr$ we define $\vdasha \gexpr \steps \reval$ iff $[], \{\} \vdasha \gexpr \steps \reval$.  
\end{definition}

\begin{figure}
  \begin{mathpar}
    \bbrule{Value}{\ 
   }{\stk,\sfrags \vdasha(\gtfun\ \ev \gtarrow \expr) \steps (\gtfun\ \ev \gtarrow \expr)^{\stk} \globs \sfrags
   }

    \bbrule{{Var Local}}{
          \stk' \in \suffixes((\expr_1\ \expr_2), \stk, \sfrags)\\
          \stk', \sfrags \vdasha \expr_2 \steps \reval_v \globs \sfrags_v
      }{((\expr_1\ \expr_2) \lC \stk),\sfrags \vdasha\ev^0 \steps \reval_v \globs \sfrags_v
      }

    \bbrule{{Var Non-Local}}{
        \stk' \in \suffixes((\expr_1\ \expr_2), \stk, \sfrags)\\
          \stk', \sfrags \vdasha\expr_1 \steps (\gtfun\ \ev_1 \gtarrow \expr )^{\stk_1} \globs \sfrags_1 \\
        \stk_1,\sfrags_1 \vdasha\ev^\dbn \steps \reval_v \globs \sfrags_v
      }{((\expr_1\ \expr_2) \lC \stk),\sfrags \vdasha\ev^{\dbn+1} \steps \reval_v \globs \sfrags_v
      }

      \bbrule{Application}{
        \stk,\sfrags \vdasha\expr_1 \steps (\gtfun\ \ev \gtarrow \expr)^{\stk_1} \globs \sfrags_1\\        
        ((\expr_1\ \expr_2) \lCk \stk), (\sfrags_1 \cup ((\expr_1\ \expr_2) \lCk \stk)) \vdasha\expr  \steps \reval_v \globs \sfrags_v
      }{\stk,\sfrags \vdasha (\expr_1\ \expr_2) \steps \reval_v \globs \sfrags_v
      }
  \end{mathpar}
  \caption{Pure function program analysis rules}
  \label{fig_steps_analysis}
\end{figure}

\subsubsection*{Explaining the rules}
The rules closely mirror the $\vdash$ rules of Figure \ref{fig_steps}, but at application time the call stack $\stk$ is pruned to be at most $k$-length.  Since these are big-step operational semantics there is no issue with call-return alignment when the stack is pruned: a big-step execution implicitly aligns calls and returns.  The only information losses are that (1) the stack may not be accurate when used for variable lookups, and (2) recursion may turn into a cycle by visiting the same stack again.  Note that for simplicity of presentation, we have removed the evaluation/discard of the function argument from the \ruleref{Application} rule.

Recall that both \ruleref{Var Local} and \ruleref{Var Non-Local} need to pop off a stack frame and perform a lookup (the argument in the former and the function in the latter).  If we naively followed this approach here, since these stacks may have been pruned, at some point we could ``pop off'' the end of the stack and would completely lose all context information.  To avoid this imprecision, in the \ruleref{Var} rules, we don't just pop a frame off the stack, we pop a frame from the top but then potentially \emph{add} a new frame to the bottom of the stack to maintain up to $k$ frames of stack context. We pick such frames based on the history of all call stack fragments seen thus far in $S$, so we only ``stitch'' together a new fragment if it aligns with the remnant of the current stack.

Delving into a bit more detail, the $\sfrags$ set can be seen to be recording each new stack fragment encountered in every  \ruleref{Application} rule: when executing the function body, we take the new stack $(\expr_1\ \expr_2) \lCk \stk$ and union it into $\sfrags$ to record it for future stitching.  Then, some \ruleref{Var} rule, say \ruleref{Var Non-Local}, pops a frame $\frm_2$ off the stack $\frm_2\lC\stk$ and finds an existing stack fragment $\stk_0 \in \sfrags$ that shares the top $k$ frames with $\frm_2\lC\stk$ \emph{and} potentially has extra frames $\stk'$ that we can use to extend the fragment when popping off $\frm_2$: we pop $\frm_2$ from the top but then add $\stk'$ to the end at the same time.  Invariably, $\stk'$ will be either 0 or 1 frames.  A similar stack stitching is performed in \ruleref{Var Local}.  We define this mechanism as $\suffixes(\frm,\stk,\sfrags) = \{\stk \lC  \stk' \mid \exists \stk'' \in \sfrags.\ \stk'' = \frm \lC \stk \lC \stk'\}$.

In order to better understand stack stitching, we provide an example scenario in Figure \ref{fig_stitching}, which corresponds to a series of variable lookups. Suppose a \ruleref{Var Non-Local} lookup started with stack $[\frm_2, \frm_1]$, then per $\suffixes(\frm_2, [\frm_1], \sfrags)$ where $\sfrags$ is defined in the figure, the only fragment $\stk'' \in \sfrags$ that has the required prefix $\frm_2 \lC\ [\frm_1] = [\frm_2, \frm_1]$ is in fact just $[\frm_2, \frm_1]$ itself. After popping the head frame, the stack will be just $[\frm_1]$.  Thus far there is no lossiness, but if there was again a \ruleref{Var Non-Local} rule firing with this stack, we can see from $\suffixes(\frm_1, [], \sfrags)$ that there are now two fragments $\stk'' \in \sfrags$ with prefix $[\frm_1] \lC\ [] = [\frm_1]$, namely $[\frm_1, \frm_0]$ and $[\frm_1, \frm_3]$.  Upon popping $\frm_1$, we obtain not an empty stack but $[\frm_0]$ and $[\frm_3]$, respectively, due to stitching with these two stacks. At most one of these is a correct fragment from the deterministic non-abstract execution; the analysis will proceed nondeterministically with one stack fragment.  If $[\frm_0]$ was used, it could further stitch with $[\frm_0, \frm_1]$, resulting in $[\frm_1]$.  This is a previously seen stack (hence the back edge), and if it was in fact a previously-visited program state then it would be stubbed/pruned in the implementation (see the following subsection for more on stubbing).  Note it can be seen from the above that if $k<2$ there would be no additional frames gained from stitching with $\sfrags$ so $k$ should \emph{de facto} be at least $2$.

\begin{figure}[h]
  \begin{tikzpicture}[
    state/.style = {anchor=west},
    trans/.style = {sloped, fill=white, inner sep=2pt},
    ->,
    grow=right,
    level 1/.style={level distance=6mm, sibling distance=10mm},
    level 2/.style={level distance=32mm},
  ]
  \node{...}
  child {
    node[state]{$[\frm_2, \frm_1]$}
    child {
      node[state](cycle){$[\frm_1]$}
      child {
        node[state]{$[\frm_3]$}
        edge from parent node [trans] {\tiny stitch w/$[\frm_1, \frm_3]$ \& pop}
      }
      child {
        node(stub)[state]{$[\frm_0]$}
        edge from parent node [trans] {\tiny stitch w/$[\frm_1, \frm_0]$ \& pop}
      }
      edge from parent node [trans] {\tiny stitch w/$[\frm_2, \frm_1]$ \& pop}
    }
    edge from parent
  };
  \path[<-, bend left=20] (cycle.north) edge node [trans] {\tiny stitch w/$[\frm_0, \frm_1]$ \& pop} (stub.north);
  \end{tikzpicture}

  \caption{Stack stitching process when $\sfrags = \{[\frm_2], [\frm_2,\frm_1], [\frm_1, \frm_0], [\frm_1, \frm_3], [\frm_0, \frm_1]\}$}
  \label{fig_stitching}
\end{figure}
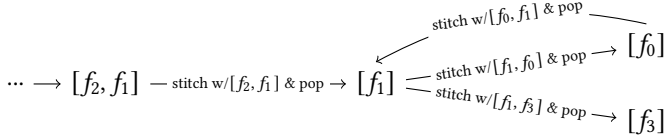

It is easy to show the analysis is computable by a raw counting argument: all data structures in the proof tree are bounded by the size of the program.
\begin{restatable}{lemma}{coreanaldec}
  For all $\gexpr$, the set $\{\reval \mid [], \emptyset \vdasha \gexpr \steps \reval\}$ is recursive.
\end{restatable}


\subsubsection*{Soundness}

A program analysis is \emph{sound} if for all programs the analysis can produce any value that the operational semantics can.  Soundness is largely straightforward here as the rules are identical besides stack pruning and stitching.  All we need to show is that the concatenation of stack fragments in $\sfrags$ will include enough pieces to glue together the actual runtime stack of a concrete run.  Since each new $k$-length fragment is added to $\sfrags$, by induction every fragment of the actual stack will be in the $\sfrags$ set in the analysis and so there will be a stitching that produces the concrete stack.

\begin{restatable}[Soundness]{theorem}{analysissoundness}
\label{lem_analysis_soundness}  
  If \ $[] \vdash \gexpr \steps (\gtfun\ \ev \gtarrow \expr)^{\stk}$ then $[], \emptyset \vdasha \gexpr \steps (\gtfun\ \ev \gtarrow \expr)^{\stk\lceil_k} \globs \sfrags$ for some $\sfrags$. 
\end{restatable}

\subsection{An All-Paths Formalization of Program Analysis}
\label{sec_all_paths_core}
The program analysis of Section \ref{sec_core_analysis} does not directly define an efficient implementation: it is both nondeterministic with multiple proof trees possible for the same program and has unbounded cycles of proof construction attempt which need to be pruned with an occurs check.  In standard presentations of higher-order program analyses that use an operational semantics \cite{MightAbstractInterpretersForFree,AAM}, the system of the previous section is as far as analyses have been specified in a rule-based system.

We show here that it is in fact possible to fully specify the program analysis solely via a proof system that returns a \emph{set} of possible results.  This gives a rule set that serves as a more direct basis for implementation.  It also constitutes a more formal specification as the entire analysis is captured in proof rules without recourse to pseudocode.  We believe this ``all paths'' pure-rule-based approach should also apply to standard program analysis presentations: it is not using any fundamental property of the pure demand approach.

Note that even without value widening there is still an analogous incremental aspect of how information is accumulated over the analyis run: multiple stack fragments will appear in the proof tree containing the same call site, each one indirectly adding potentially new information on what arguments are passed to the function since the argument lookup may be different in different fragments.

\begin{wrapfigure}{r}{.38\textwidth}%
    \begin{grammar}
            \grule[visited states]{\visited}{\{ \la \cycmark,\stk,\sfrags\ra, \dots \} }                    
             \grule[results]{\reval}{
              \{ (\gtfun\ \ev \gtarrow \expr)^{\stk}, \ldots \} 
  }
            \grule[cycle mark]{\cycmark}{
          \ev^\dbn \gor (\expr \ \expr)
}
            \grule[Seen stacks]{\sfrags}{
                        \{\stk, \ldots \}
            }
        \end{grammar}
    \caption{Stuctures for all-paths analysis}
    \label{fig_grammar_core_all_paths}
\end{wrapfigure}

We now present this implementation-based rule system, which we term an \emph{all-paths} presentation because the rules run all possible combinations of values in parallel. Along with running all paths, we also track and prune potentially cyclical proofs in the rules themselves by maintaining a set of visited states $\visited$ in the rules and using special \ruleref{Stub} rules to prune proof goals that are provably cyclic.

Figure \ref{fig_grammar_core_all_paths} presents the changed grammar elements needed for the all-paths analysis.  Value results $\reval$ are now \emph{sets} of function values $\{ (\gtfun\ \ev_1 \gtarrow \expr_1)^{\stk_1}, \ldots, (\gtfun\ \ev_n \gtarrow \expr_n)^{\stk_n}\}$ where each value represents one possible result of the analysis.  The set can also be empty for the case that there is a cycle (meaning the program diverges).  We use visited sets $\visited = \{\la \cycmark,\stk, \sfrags \ra, \dots \}$ in judgements to enforce an occurs check: we keep an audit trail that includes both every function call frame as well as every variable lookup request ($\cycmark$ in $\la \cycmark,\stk, \sfrags \ra$ is either a $(\expr_1\ \expr_2)$ or an $\ev^\dbn$).  Additionally, in this set we include the complete snapshot of the state, \ie the current $\stk$ and $\sfrags$.  We will use this audit trail in $\visited$ to see if we are asking the \emph{exact} same question previously asked in a parent node in the proof tree, and will ``stub'' out any proof goal if so as it is provably cycling.  

We now define the all-paths analysis relation, $\vdashaa$, with the rules of Figure \ref{fig_steps_analysis_all_paths}.  There, we overload $\bigcup$ on a set of pairs of sets to mean $\bigcup$ applied to each component of the pair and coalesced back as a pair.

\begin{definition}
  \label{def_core_all_paths_analysis}
  Relation $\stk, \sfrags, \visited \vdashaa \expr \steps \reval \globs \sfrags$ is defined by the rules of Figure \ref{fig_steps_analysis_all_paths}.
\end{definition}

Observe that \ruleref{Application} places the current state in the stack $\visited_\textrm{new}$ that then gets passed up the proof tree. Then, \ruleref{Application Stub} checks if the current state is already in $\visited$ and if so it stubs the proof with a leaf node.  The rules need to explicitly perform this stubbing or any program that contains unbounded recursion will not have any (well-founded) proof tree.  Along with stubbing out cycles, the analysis takes all paths in parallel; the rules realize this by returning sets of values. 

\begin{figure}
  \begin{mathpar}

   \bbrule{Value}{\ 
   }{\stk, \sfrags, \visited \vdashaa \gtfun\ \ev \gtarrow \expr \steps \{(\gtfun\ \ev \gtarrow \expr)^{\stk}\} \globs \sfrags
   }

  \bbrule{Var Local}{
        \stk_0 = ((\expr_1\ \expr_2) \lC \stk)\\
        \la\ev^0, \stk_0,\sfrags\ra \not\in \visited\\
        \stks = \suffixes((\expr_1\ \expr_2),\stk,\sfrags)\\
        (\reval_1,\sfrags_1) = \bigcup \{ (\reval_0,\sfrags_0) \mid \exists \stk_1 \in \stks.\  \stk_1, \sfrags, \visited \cup \{\la\ev^0, \stk_0, \sfrags\ra\}\vdashaa \expr_2  \steps \reval_0 \globs \sfrags_0 \}
   }{\stk_0, \sfrags, \visited \vdashaa \ev^0  \steps \reval_1 \globs \sfrags_1
   }

    \bbrule{Var Non-Local}{
        \stk_0 = ((\expr_1\ \expr_2) \lC \stk)\\
        \la\ev^\dbn, \stk_0, \sfrags\ra \not\in \visited\\
        \stks = \suffixes((\expr_1\ \expr_2),\stk,\sfrags)\\
        (\reval_1,\sfrags_1) = \bigcup \{ (\reval_0,\sfrags_0) \mid \exists \stk_1 \in \stks.\ 
        \stk_1, \sfrags, \visited \cup\{\la\ev^{\dbn+1}, \stk_0, \sfrags\ra\}  \vdashaa \expr_1 \steps \reval_0 \globs \sfrags_0 \}\\
        (\reval_2,\sfrags_2) = \bigcup \{ (\reval_0',\sfrags_0') \mid \exists \expr_f^{\stk_1} \in \reval_1.\ \stk_1, \sfrags_1, \visited \cup\{\la\ev^{\dbn+1}, \stk_0, \sfrags_1\ra\}  \vdashaa \ev^\dbn \steps \reval_0' \globs \sfrags_0'\}\\
      }{\stk_0, \sfrags, \visited \vdashaa \ev^{\dbn+1}  \steps \reval_2 \globs \sfrags_2
      }

      \bbrule{Var Stub}{
         \la\ev^\dbn,\stk, \sfrags\ra \in \visited\\
     }{\stk, \sfrags, \visited \vdashaa \ev^\dbn  \steps \emptyset \globs \sfrags
     }     

      \bbrule{Application}{
        \frm = (\expr\ \expr')\\
        \stk, \sfrags,  \visited \vdashaa \expr \steps \reval_1 \globs \sfrags_1\\
        \visited_{\textrm{new}}=\{\la \frm, \stk, \sfrags_1\cup \{\frm \lCk \stk\}\ra\}\\
        \visited_{\textrm{new}}\not\subseteq \visited\\
        {\begin{minipage}{4.5in}
          $(\reval_2,\sfrags_2) = \bigcup \Bigl\{ (\reval_0,\sfrags_0) \mid  (\gtfun\ \ev_1 \gtarrow \expr_1)^{\stk_1} \in \reval_1 \land \\
        \hspace*{6em} (\frm \lCk \stk), \sfrags_1 \cup \{\frm \lCk \stk\}, \visited \cup \visited_{\textrm{new}} \vdashaa \expr_1 \steps \reval_0 \globs \sfrags_0 \Bigr\}
        $\end{minipage}}
      }{\stk, \sfrags, \visited \vdashaa \frm \steps \reval_2 \globs \sfrags_2
      }

      \bbrule{App Stub}{\la (\expr\ \expr'), \stk, \sfrags\ra \in \visited
      }{\stk, \sfrags, \visited \vdashaa (\expr\ \expr') \steps \emptyset  \globs \sfrags
      }
   
  \end{mathpar}
  \caption{Core language all-paths program analysis rules}
  \label{fig_steps_analysis_all_paths}
\end{figure}

The all-paths system $\vdashaa$ is just merging all the possible nondeterministic proof trees of the single-path system $\vdasha$ into a single deterministic tree. Note that it treats the $\sfrags$ set slightly differently as it accumulates one set for all possible paths while the single-path rules accumulate an $\sfrags$ on a more accurate per-path basis.  This is analogous to the single-threaded store optimization of conventional analyses \cite{ShiversThesis}. In both our system and the conventional systems there is some loss of precision with a single-threaded store, but in return a significant speed-up is gained.   

We conjecture that the all-paths system produces in one go at least all the results of the per-path system (it may produce more results due to the potentially larger $\sfrags$).
\begin{restatable}[All-Paths Soundness]{conjecture}{aisaa}
  Letting $\reval$ be $\{(\gtfun\ \ev_1 \gtarrow \expr_1)^{\stk_1}, \ldots, (\gtfun\ \ev_n \gtarrow \expr_n)^{\stk_n}\}$, it is the case that for all $1\leq i \leq n$, if $[], \emptyset \vdasha \expr \steps (\gtfun\ \ev_i \gtarrow \expr_i)^{\stk_i} \globs \sfrags_i$ holds and $\expr$ has no other values under $\vdasha$, then $[], \emptyset \vdashaa \expr \steps \{(\gtfun\ \ev_1 \gtarrow \expr_1)^{\stk_1}, \ldots, (\gtfun\ \ev_n \gtarrow \expr_n)^{\stk_n},\ldots\} \globs \sfrags'$.
\end{restatable}

\begin{lemma}[Termination]
  Fixing some $k \geq 1$, for all $\gexpr$, $[], \emptyset \vdashaa \gexpr \steps \reval$.
\end{lemma}
The above language with only functions leaves out several aspects of program analyses, namely how they approximate atomic data such as integers, and how they approximate recursive data structures such as trees.  We address these dimensions in the following section.

\subsection{An All-Paths Analysis for an Extended Language}
\label{sec_full_analysis}
In this section, we extend the core all-paths analysis of the previous section to the extended language grammar of Section \ref{sec_extended_language}.

Figure \ref{fig_Grammar_Extended_Analysis} presents the new language constructs for the analysis, extending or modifying those of Figure \ref{fig_grammar_extended}.  Results $\reval$ are now sets of more than just functions, they are sets of \emph{value atoms} $\revalz$ which includes records, functions, operators, etc. Additionally, the components of value atoms themselves may have multiple possibilities, so \eg the components of a binary operator $\revalz = \reval_1 \binop \reval_2$ are $\reval$ which are sets.  So for example  $\revalz = \{1\}+ \{4, (\{ 6 \} + \{ 7, 2 \})\}$ is possible and represents $\{5, 14, 9\}$.  In some of the examples below, we will pun and leave off the $\{\}$ on singleton sets.

\begin{wrapfigure}{r}{.62\textwidth}%
    \begin{grammar}
            \grule[value atoms]{\revalz}{
              (\reval \binop \reval) 
              \gor \reval \gtdot \lbl 
              \gor \lbl \gtin \reval
              \gor \reval^{\la\cycmark,\stk\ra}
              \gor \pathc \Vdash \reval
              \gline
              \gor (\gtfun\ \ev \gtarrow \expr)^{\stk}
              \gor \ob \lbl = \reval; \dots ; \lbl = \reval \cb
              \gor \stub^{\la\cycmark,\stk\ra}
              \gor \eint
              \gor \ebool
             }
             \grule[value results]{\reval}{
              \{ \revalz,\ldots,\revalz\}
             }
             \grule[cycle mark]{\cycmark}{
              \ev^\dbn \gor (\expr \ \expr)
    }
             \grule[path condition]{\pathc}{
              \reval = \ebool
             }
             \grule[visited state sets]{\visited}{ 
              \{ \la \cycmark,\stk,\sfrags\ra, \dots \}
             }
             \grule[Seen stacks]{\sfrags}{
              \{\stk, \ldots \}
             }
        \end{grammar}
    \caption{Extended language analysis grammar}
    \label{fig_Grammar_Extended_Analysis}
\end{wrapfigure}

\subsubsection*{Results That Are Recurrences}
There is a surprisingly natural manner in which recurrences over integer and record data can be extracted as results $\reval$. Arithmetic recurrences can then be translated to Constrained Horn Clauses (CHCs) and a solver can verify them.  We know of no existing analysis that also extracts recurrences, and the demand aspect of our formalism is one reason why it is so natural here.  We now outline informally how it works.

The first step is to label result values in the variable and application rules with their state, \ie their expression and current stack: $\reval^{\la\cycmark,\stk\ra}$.  If the analysis visits the same state $\la\cycmark,\stk\ra$ again deeper in the proof tree, we don't just stub the result and return an empty set as in the core system, we return an explicit \emph{stub value} $\stub^{\la\cycmark,\stk\ra}$ that is a recursive reference to the parent node also marked with $\la\cycmark,\stk\ra$.   With nondeterministic result sets plus recursive values, the analysis can infer unbounded ranges of data.  For example, the result value $\{ \texttt{0}, (\texttt{1} + \stub^{\la\cycmark,\stk\ra})\}^{\la\cycmark,\stk\ra}$ denotes a recurrence: the outer label ${\la\cycmark,\stk\ra}$ denotes the first time the state $\la\cycmark,\stk\ra$ is hit, and $\stub^{\la\cycmark,\stk\ra}$ denotes a stubbed recursive reference back to the top of this result because it is also labeled with the same $\la\cycmark,\stk\ra$.  The mathematical meaning of such a recurrence follows other forms of recursive definition: it is the unrolling of the recurrence, in this case $\{ \texttt{0}, (\texttt{1} + \{\texttt{0} , (\texttt{1} + \{\texttt{0}, (\texttt{1} + \dots)\})\})\}$, \ie here it is the non-negative integers.   Similarly, $\{ \texttt{0}, (2 + \stub^{\la\cycmark,\stk\ra})\}^{\la\cycmark,\stk\ra}$ would express the non-negative evens, and $\{ \ob \texttt{nil = }\ob\cb\cb, (\ob \texttt{hd = \texttt{1}, tl = }\stub^{\la\cycmark,\stk\ra}\cb)\}^{\la\cycmark,\stk\ra}$ would express an arbitrary-length list of \texttt{1}'s.  

We will define a function $\evalval$ in Definition \ref{def_evalval} below that extracts the actual values from these recursive result sets, for example $\evalval(\{ \texttt{0}, (\texttt{1} + \stub^{\la\cycmark,\stk\ra})\}^{\la\cycmark,\stk\ra})=\{0, 1, 2, \ldots\}$.  In the general case, such result values could be complex arithmetic recurrence relations that could be uncomputable.  In our implementation, we will compile them to Constrained Horn Clauses and automatically solve them using a CHC solver. 

To add one more level of precision, we also include path sensitivity in this system.  It can both serve to avoid executing conditional code for which the guard provably fails, and to allow even more precise results to be inferred.  These result values are of the form $\pathc \Vdash \reval$, meaning that if $\pathc$ holds then the result is $\reval$ and no result here if the condition fails. 

\subsubsection*{The Proof Rules}

Before defining the proof rules we will define the auxiliary relation $\evalval$ alluded to above to solve result values, producing a (possibly infinite) flat set.  We use $\evalval$ in the rules when concrete details of the result affect the next control step; this is only required for boolean conditions and function values.  Note that the implementation will use a (quickly) computable approximation to this uncomputable function.

\begin{definition}[Result Evaluation]
  \label{def_evalval}
  Define $\evalval(\reval)$ for the analysis to be $\evalval(\reval,\emptyset)$, which in turn is defined as the least fixed point of the monotone functional of Figure \ref{fig_evalvalanal}.
$\evalval$ is undefined for cases of type mismatch, \eg projecting a number or adding records.
\end{definition}

\begin{figure}
$$\begin{array}{rcl}
  \evalval(\{\revalz_1,\dots,\revalz_n\}, \map) & = & \bigcup_{1\leq i\leq n} \evalval(\revalz_i, \map)\\[1.2ex]
  \evalval(\reval^{\la \cycmark,\stk\ra}, \map) & = & \evalval(\reval,\map \cup \{\la \cycmark,\stk\ra \mapsto \reval\})\\
  \evalval(\stub^{\la \cycmark,\stk\ra}, \map) & = & 
    \begin{cases} \evalval(\map(\la \cycmark,\stk\ra), \map)\text{, provided }\map(\la \cycmark,\stk\ra) \text{ is defined}\\
    \{\eval \mid \eval \text{ is a value} \} \text{ otherwise}
  \end{cases}\\
  \evalval(\pathc \Vdash \reval, \map) & = & 
  \begin{cases}
     \evalval(\reval, \map)\text{, provided } \gttrue \in \evalval(\pathc,\map)\\
     \{\} \text{ otherwise}
  \end{cases}\\
  \evalval(\reval_1 \binop \reval_2, \map) & = & \{ \eval \mid \eval_1 \in \evalval(\reval_1, \map) \text{ and } \eval_2 \in \evalval(\reval_2, \map) \text{ and } \eval = \eval_1 \binop \eval_2\}\\
  \evalval(\reval \gtdot \lbl_i, \map) & = & \{ \eval_i \mid \ob \lbl_1 \gteq \eval_1; \dots; \lbl_i \gteq \eval_i; \dots;\lbl_n \gteq \eval_n \cb \in \evalval(\reval, \map)\}\\  
  \evalval(\lbl \gtin \reval, \map) & = & \{ \ebool \mid \ob \lbl_1 \gteq \eval_1; \dots;\lbl_n \gteq \eval_n \cb \in \evalval(\reval, \map)\text{ and } \\
  &&\ebool = \gttrue \text{ iff } \lbl \in \{\lbl_1,\dots,\lbl_n\} \}\\
  \evalval(\ob \lbl_1 \gteq \reval_1; \dots; \lbl_n \gteq \reval_n \cb, \map) & = & \{ \ob \lbl_1 \gteq \eval_1; \dots; \lbl_n \gteq 
  \eval_n \cb \mid \forall 1 \leq i \leq n.\ \eval_i \in \evalval(\reval_i, \map)\}\\
  \evalval((\gtfun\ \ev \gtarrow \expr, \map)^{\stk}, \map) & = & \{(\gtfun\ \ev \gtarrow \expr, \map)\}\\  
  \evalval(n, \map) & = & \{ n\}\\
  \evalval(\ebool, \map) & = & \{ \ebool\}
\end{array}
$$\caption{The definition of $\evalval$ for the analysis}\label{fig_evalvalanal}
\end{figure}

The $M$ mapping in the figure caches all the recurrence parents and the $\stub$ evaluation unrolls the recurrence.  For a final program result, all stubs will have a parent in the tree, but when evaluating conditional expressions mid-evaluation some stubs may not yet have parents; in this case the value is arbitrary.  Note that there could be unbounded cycles so we specify we are taking the least fixed point of the above monotone functional.  For this reason, $\evalval$ is not computable.  We will cover how we approximate $\evalval$ when we describe our implementation.
We may now define the extended all-paths analysis rules.

\begin{definition}
  \label{def_extended_all_paths_analysis}
  Relation $\stk, \sfrags, \visited \vdashea \expr \steps \reval \globs \sfrags$ is defined by the rules of Figure \ref{fig_steps_analysis_extended}.
\end{definition}

\begin{figure}
  \begin{mathpar}

   \bbrule{Value}{\eval \text{ is not a function}
   }{\stk, \sfrags, \pathc, \visited \vdashea \eval \steps \{\eval\} \globs \sfrags
   }

   \bbrule{Value Fun}{\ 
   }{\stk, \sfrags, \pathc, \visited \vdashea \gtfun\ \ev \gtarrow \expr \steps \{(\gtfun\ \ev \gtarrow \expr)^{\stk} \}\globs \sfrags
   }

   \bbrule{Var Local}{
    \stk_0 = ((\expr_1\ \expr_2) \lC \stk)\\
    \la\ev^0, \stk_0,\sfrags\ra \not\in \visited\\
    \stks = \suffixes((\expr_1\ \expr_2),\stk,\sfrags)\\
    (\reval_1,\sfrags_1) = \bigcup \{ (\reval_0,\sfrags_0) \mid \exists \stk_1 \in \stks.\  \stk_1, \sfrags, \pathc, \visited \cup \{\la\ev^0, \stk_0,\sfrags\ra\}\vdashea \expr_2  \steps \reval_0 \globs \sfrags_0 \}
    }{\stk_0, \sfrags, \pathc, \visited \vdashea \ev^0  \steps \{\reval_1^{\la\ev^0, \stk_0\ra}\} \globs \sfrags_1
    }

      \bbrule{Var Non-Local}{
        \stk_0 = ((\expr_1\ \expr_2) \lC \stk)\\
        \la\ev^{\dbn+1}, \stk_0,\sfrags\ra \not\in \visited\\
        \stks = \suffixes((\expr_1\ \expr_2),\stk,\sfrags)\\
        (\reval_1,\sfrags_1) = \bigcup \{ (\reval_0,\sfrags_0) \mid \exists \stk_1 \in \stks.\ 
        \stk_1, \sfrags, \pathc, \visited \cup \{\la\ev^{\dbn+1}, \stk_0,\sfrags\ra\} \vdashea \expr_1 \steps \reval_0 \globs \sfrags_0 \}\\
        (\reval_2,\sfrags_2) = \bigcup \{ (\reval_0',\sfrags_0') \mid \exists \expr_f^{\stk_1} \in \evalval(\reval_1).\ \stk_1, \sfrags_1, \pathc, \visited \cup \{\la\ev^{\dbn+1}, \stk_0,\sfrags_1\ra\} \vdashea \ev^{\dbn} \steps \reval_0' \globs \sfrags_0'\}\\
      }{\stk_0, \sfrags, \pathc, \visited \vdashea \ev^{\dbn+1}  \steps \reval_2 \globs \sfrags_2
      }

      \bbrule{Var Stub}{
         \la\ev^\dbn,\stk,\sfrags\ra \in \visited\\
     }{\stk, \sfrags, \pathc, \visited \vdashea \ev^\dbn  \steps \{ \stub^{\la\ev^\dbn,\stk\ra} \} \globs \sfrags
     }
         
      \bbrule{Operation}{
        \stk, \sfrags, \pathc, \visited \vdashea \expr_1 \steps \reval_1  \globs \sfrags_1\\
        \stk, \sfrags_1, \pathc, \visited \vdashea \expr_2 \steps \reval_2  \globs \sfrags_2
      }{\stk, \sfrags, \pathc, \visited \vdashea \expr_1 \binop \expr_2 \steps \{\reval_1 \binop \reval_2\} \globs \sfrags_2
      }

      \bbrule{Record Value}{
        \forall i \leq n. \ \stk, \sfrags, \pathc, \visited \vdashea \expr_i \steps \reval_i \globs \sfrags_i
      }{\stk, \sfrags, \pathc, \visited \vdashea  \ob \lbl_1 \gteq \expr_1; \dots; \lbl_n \gteq \expr_n \cb \steps \{ \ob \lbl_1 \gteq \reval_1; \dots; \lbl_n \gteq \reval_n \cb \} \globs \bigcup_i \sfrags_i
      }

      \bbrule{Record Project}{
        \stk, \sfrags, \pathc, \visited \vdashea  \expr \steps \reval_0 \globs \sfrags_1\\
      }{\stk, \sfrags, \pathc, \visited \vdashea  \expr\gtdot \lbl \steps \{\reval_0 \gtdot \lbl \} \globs \sfrags_1
      }

      \bbrule{Record Inspect}{
        \stk, \sfrags, \pathc, \visited \vdashea  \expr \steps \reval_0 \globs \sfrags_1\\
      }{\stk, \sfrags, \pathc, \visited \vdashea  \lbl \gtin \expr \steps \{\lbl \gtin \reval_0 \} \globs \sfrags_1
      }

      \bbrule{Conditional}{
        \stk, \sfrags, \pathc, \visited \vdashea \expr \steps \reval_\textrm{cond} \globs \sfrags_0\\
        \ebools = \evalval(\pathc \Vdash\reval_\textrm{cond})\\
        \forall \ebool \in \ebools.\ \stk, \sfrags, (\reval_\textrm{cond} = \ebool), \visited \vdashea \expr_{\ebool} \steps \reval_\ebool \globs \sfrags_{\ebool}\\
      }{\stk, \sfrags, \pathc, \visited \vdashea \expr \gtquestion \expr_{\gttrue} \gtcolon \expr_{\gtfalse} \steps \bigcup_{\ebool \in \ebools}\{\reval_\textrm{cond} = \ebool \Vdash \reval_\ebool\}  \globs \cup_{\ebool \in \ebools} \sfrags_{\ebool} \cup \sfrags_0
      }

      \bbrule{Application}{
        \frm = (\expr\ \expr')\\
        \stk, \sfrags, \pathc, \visited\vdashea \expr \steps \reval_1 \globs \sfrags_1\\
        \visited_{\textrm{new}}=\{\la \frm, \stk,\sfrags_1\cup \{\frm \lCk \stk\}\ra\}\\
        \visited_{\textrm{new}}\not\subseteq \visited\\
        {\begin{minipage}{4.5in}
          $(\reval_2,\sfrags_2) = \bigcup \Bigl\{ (\reval_0,\sfrags_0) \mid  (\gtfun\ \ev_1 \gtarrow \expr_1)^{\stk_1} \in \evalval(\reval_1) \land \\
        \hspace*{6em} (\frm \lCk \stk), \sfrags_1 \cup \{\frm \lCk \stk\}, \visited\cup \visited_{\textrm{new}}\vdashea \expr_1 \steps \reval_0 \globs \sfrags_0 \Bigr\}
        $\end{minipage}}
      }{\stk, \sfrags, \pathc, \visited \vdashea (\expr\ \expr') \steps \{\reval_2^{\la \frm, (\frm \lCk \stk)\ra}\} \globs \sfrags_2
      }

      \bbrule{App Stub}{\la (\expr\ \expr'), \stk, \sfrags\ra \in \visited
      }{\stk, \sfrags, \pathc, \visited \vdashaa (\expr\ \expr')  \steps \{\stub^{\la(\expr\ \expr'),\stk\ra}\}  \globs \sfrags
      }
   
  \end{mathpar}
  \caption{Extended language program analysis rules}
  \label{fig_steps_analysis_extended}
\end{figure}

The rules follow the core all-paths analysis of Figure \ref{fig_steps_analysis_all_paths} for function application and variable lookup; note that we do not need to stub the new operations as all unbounded control flows must pass through variable lookup and function application.  The primary difference now is while in the core system we could just return an empty set of results when the proof was cycling, here we need to not simply prune away such sub-derivations, we instead need to return a stub $\stub^{\la\cycmark,\stk\ra}$ to explicitly mark the point that the derivation cycled and to allow us to extract recurrences.

The \ruleref{Conditional} rule deserves some explanation.  The set $\ebools$ is a subset of $\{\gttrue,\gtfalse\}$ and is the result of $\evalval$-uating the conditional result to extract which condition(s) may hold.  Only the conditional branches in $\ebools$ will fire, giving path-sensitivity. The path conditions $\reval_\textrm{cond} = \ebool$ are added to the assumptions indicating the constraint placed on $\reval_\textrm{cond}$.  Note that the path condition is only useful in the case that $\ebools$ is not a singleton: if $\ebools$ is say $\{\gttrue\}$ for example then the constraint will be $\reval_\textrm{cond} = \gttrue$ which is direct from $\evalval(\pathc \Vdash \reval_\textrm{cond}) = \{\gttrue\}$.  Path constraints $\reval_\textrm{cond} = \ebool \Vdash \reval_\ebool$ are added to the results to allow us to extract more precise recurrences.



The classic program analysis for higher-order functional languages is $k$CFA \cite{ShiversThesis}, where $k$ is the number of frames of call stack preserved.  It is known to be exponential for $k\geq 1$ \cite{kCFAEXPTIME}.  $1$CFA is exponential because even with only one stack frame of reference the propagation of the environment forward will be enough to cause a chain of $n$ calls to have $2^n$ branches that need to be explored by the analysis.  We similarly may obtain a blowup of branches due to the number of different stack stitchings possible, and so our algorithm also appears to be exponential-time. In the next section, we will cover how the analysis performs in practice.

\section{Implementation}
\label{sec_implementation}

In this section, we describe the implementation of the interpreter and the program analysis.  All implementations are in OCaml.  A detailed walk-through of the implementation code can be found in \infull{Appendix \ref{sec_implement_appendix}}\inshort{Appendix B of the Supplementary Material}.

\subsection{The Interpreter Implementation}
\label{sec_interp_impl}
We have implemented an interpreter for the operational semantics of the extended language of Section \ref{sec_extended_language} that passes a suite of tests, providing further evidence of the soundness of the system beyond Lemma \ref{thm_v_e}.  Our implementation includes a cache mapping $(\stk,\expr) \mapsto \reval$ to avoid repeating the same lookup question, and makes the runtime complexity roughly comparable to a standard interpreter implementation.  Our interpreter is only a proof of concept to verify correctness and is not likely useful as an implementation methodology due to the potential space blowup.

\subsection{The Program Analysis Implementation}

Our primary implementation effort has been the program analysis for the extended language found in Section \ref{sec_full_analysis}.  We report here on the current state of our proof-of-concept implementation.

We first describe how we simplify results $\reval$ and how arithmetic and boolean recurrences are solved by the Spacer CHC solver \cite{SpacerTutorial22}. 

\subsubsection{Simplification of Results}
\label{sec_impl_simpl_results}


Simplification is important for several reasons.  Along with the fact that the unsimplified results are unreadable, the CHC solver described in the next subsection will be many times more efficient if fed simplified clauses.  So, we have implemented a simplification function in OCaml that takes in an $\reval$ and produces an equivalent but simplified version $\reval'$.  Simplification is also critical because, while Spacer can deal with constraints involving integers or booleans, recursive records are not handled by Spacer and so we also need to perform record projections appearing in $\reval$. The challenge of simplification is to avoid unbounded computation; the $\evalval$ function of Section \ref{sec_full_analysis} is uncomputable due to the unbounded recurrences.  

We implement simplification as a term rewriting system that includes the following rewrites:

\begin{enumerate}
    \item Non-recurrence simplification: $\reval$ simplifies to $\evalval(\reval)$ if $\reval$ contains no stubs 
    \item Constant folding: \eg \texttt{2 + 5} simplifies to \texttt{7}
    \item Associativity: \eg \texttt{1 + (2 - $\stub^{\la\cycmark,\stk\ra}$)} simplifies to \texttt{(1 + 2) - $\stub^{\la\cycmark,\stk\ra}$}
    \item Record projection: \eg \camlil!{hd = $$ 1, tl = $$ {nil = $$ 0$$}}.hd! simplifies to \texttt{1}
    \item Unrolling of recurrences: $(\stub^{\la\cycmark,\stk\ra})\gtdot \lbl$ expands to $\reval_0 \gtdot \lbl$ if $\reval_0$ has label $\la \cycmark,\stk\ra$ in the original $\reval$
    \item Distribution of operations into sets: \eg $\{$\texttt{2, 4$\}$ + $\{$1, 3}$\}$ simplifies to \texttt{$\{$2 + 1, 2 + 3, 4 + 1, 4 + 3$\}$}
    \item Merging of cycles: \eg $\reval_1^{\la\cycmark,\stk\ra} + \reval_2^{\la\cycmark,\stk\ra}$ simplifies to $(\reval_1 + \reval_2)^{\la\cycmark,\stk\ra}$ 
    \item Merging of path conditions: \eg $(\pathc\Vdash\reval_1) + (\pathc\Vdash\reval_2)$ simplifies to $\pathc\Vdash(\reval_1 + \reval_2)$ 
\end{enumerate}

We limit unrolling of recurrences to only the case where the result of $(\stub^{\la\cycmark,\stk\ra})\gtdot \lbl'$ or $\lbl' \gtin(\stub^{\la\cycmark,\stk\ra})$ is required (to obtain either a function or boolean result for \ruleref{Conditional}, \ruleref{App}, or \texttt{letassert}).  In order to avoid unbounded expansion, we mark each syntactic $\expr \gtdot \lbl$ and only expand each projection occurrence once.  This will allow, \eg \texttt{(($\stub^{\la\cycmark,\stk\ra}$).tl).hd} to simplify to \texttt{1} in a context of a parent label $(\ob \texttt{hd = \texttt{1}, tl = }\stub^{\la\cycmark,\stk\ra}\cb)^{\la\cycmark,\stk\ra}$ since each projection is syntactically distinct: we can unroll the record twice.  Associativity can only be applied in one direction to avoid nontermination.


The simplifier still needs refining and is one reason why the analysis is sometimes slow.

\subsubsection{Evaluating Arithmetic Recurrences via CHCs}
\label{sec_eval_recur_chcs}

Integer recurrences are not generally decidable, but a heuristic solver for these recurrences can be constructed by translating the $\evalval$ definition for a fixed $\reval$ to a set of Constrained Horn Clauses (CHCs).  We define the translation function here; in the implementation, we use Z3's CHC solver, Spacer \cite{SpacerTutorial22}, to determine if the clauses produced by our translator have a solution.  Spacer is just a heuristic that may not terminate, so it needs to be run with a timeout, and in the case of timeout we must fall back to an unconstrained result such as all integers.

Assume given a fixed $\reval$.  We assume there is a function $\id(\reval_0)$ that, for each non-stub sub-expression occurrence $\reval_0$ in $\reval$, returns a unique ID $\glab$.  For the parent and stub labeled nodes, we require $\id(\reval_0^{\la \cycmark,\stk\ra})=\id(\stub^{\la \cycmark,\stk\ra})$. We use notation $\vec\pathc$ for a vector of path conditions, $\vec\pathc = (\pathc_1, \dots, \pathc_n)$.  We notate truncation of vectors as $(\pathc_1, \dots, \pathc_m,\ldots,\pathc_n)\lceil^m = (\pathc_1, \dots, \pathc_m)$.

\begin{definition}[$\tochc$]
  $\tochc(\reval,\vec{\pathc},P)$ is defined by cases on $\reval$/$\revalz$ in Figure \ref{fig_tochc}.  
    The result is undefined for cases of type mismatch, \eg adding booleans.  It is also undefined for records and record operations as the CHC clauses are only used for numerical and boolean constraints.  We define the top-level $\tochc(\reval)$ as 
    $$\tochc(\reval)
    = \{\forall x_{\id(\reval_1)}, \ldots, x_{\id(\reval_n)}.\phi \mid \phi \in \tochc(\reval,\gttrue,\emptyset) \land {\sc FV}(\phi)=x_{\id(\reval_1)}, \ldots, x_{\id(\reval_n)}\}$$
\end{definition}

\begin{figure}   
    $$\begin{array}{rcl}
      \tochc(\reval,\vec{\pathc},P) & = & \bigcup_{1\leq i \leq n}\{  X_{\id(\revalz_i)}(x_{\id(\revalz_i)}) \land \pathcond(\vec\pathc) \Rightarrow X_{\id(\reval)}(x_{\id(\revalz_i)})\}\cup \\ && \bigcup_{1\leq i\leq n} \tochc(\revalz_i,\vec{\pathc},P)\\
      &&\text{ where }\reval = \{\revalz_1,\dots,\revalz_n\}\\[1.5ex]
      \tochc(\reval^{\la \cycmark,\stk\ra}, \vec{\pathc},P) & = & \tochc(\reval,\vec{\pathc},P\cup \{\la \cycmark,\stk\ra \})\\
      \tochc(\stub^{\la \cycmark,\stk\ra}, \vec{\pathc},P) & = & 
      \begin{cases} \{\} \text{ provided } \la \cycmark,\stk\ra \in P\\
      \{ X_{\la \cycmark,\stk\ra}(x_{\la \cycmark,\stk\ra}) \} \text{ otherwise}
    \end{cases}\\
      \tochc(\stub^{\la \cycmark,\stk\ra}, \vec{\pathc},P) & = & \{\}\\
      \tochc(\reval = \ebool \Vdash \reval_0, \vec{\pathc},P) & = &  \tochc(\reval_0,(\vec{\pathc}, \reval = \ebool)) \cup \tochc(\reval,\vec{\pathc},P)\\
      \tochc(\reval_1 \binop \reval_2, \vec{\pathc},P) & = &  = \{X_{\id(\reval_1)}(x_{\id(\reval_1)}) \land X_{\id(\reval_2)}(x_{\id(\reval_2)}) \land \pathcond(\vec\pathc) \Rightarrow \\
      && X_{\id(\reval)}(x_{\id(\reval_1)} \binop x_{\id(\reval_2)})\}\cup \tochc(\reval_1,\vec{\pathc},P) \cup \tochc(\reval_2,\vec{\pathc},P)\\
      %
      %
      %
      \tochc((\gtfun\ \ev \gtarrow \expr)^{\cycmark,\stk}, \vec{\pathc},P) & = & \{\}\\  
      \tochc(n, \vec{\pathc},P) & = & \{ \pathcond(\vec\pathc) \Rightarrow X_{\id(\revalz)}(n)\}\text{ where }\revalz = n\\
      \tochc(\ebool, \vec{\pathc},P) & = & \{ \pathcond(\vec\pathc) \Rightarrow X_{\id(\revalz)}(\ebool)\}\text{ where }\revalz = \ebool\\
      \pathcond(\reval_1 = \ebool_1, \dots, \reval_n = \ebool_n) & = & X_{\id(\reval_1)}(x_{\id(\reval_1)}) \land \ldots \land X_{\id(\reval_n)}(x_{\id(\reval_n)}) \land \\
      && (x_{\id(\reval_1)} \gteq \ebool_1) \land \ldots \land (x_{\id(\reval_n)} \gteq \ebool_n)\\
    \end{array}
    $$
    \caption{$\tochc$ translation of results $\reval$ to Constrained Horn Clauses}
    \label{fig_tochc}
  \end{figure}

In this definition, the $x_{\id(\reval_i)}$'s are the logic variables for graph point $\reval_i$ and the $X_{\id(\reval_i)}$'s are the corresponding predicates for what values those points can take on: for example, $X_{\id(\reval_i)}(5)$ would mean that sub-expression $\reval_i$ has value 5. 
A key aspect of the translation is how recurrences will translate to self-referential clauses; we do so by assigning the parent and stub the same label so they use the same underlying predicate.   

The $\pathc \Vdash$ in the $\reval$ guard the return values with the condition they were under; they are added as preconditions in the $\tochc$ clauses, which will allow the solver to take them into account.

Given this translation, to approximate $\evalval(\reval_\textrm{cond})$ in the \ruleref{Conditional} rule, two CHC queries should be performed.  Along with the above constraints $\tochc(\reval_\textrm{cond})$, we first issue a query with the additional clause $X_{\id(\reval_\textrm{cond})}(\gtfalse) \Rightarrow \gtfalse$; if this query is satisfiable, that means we are allowed to only evaluate the $\gttrue$ branch as the $\gtfalse$ branch can never hold.  If this query is unsatisfiable, all it means is we cannot only evaluate the $\gttrue$ branch.  The check as to whether we are allowed to only evaluate the $\gtfalse$ branch is the mirror image.


\subsubsection{Overview of the Analysis Implementation}
\label{sec_overview_analysis_impl}

Our implementation largely aims to follow the rules of Section \ref{sec_full_analysis}, replacing the uncomputable $\evalval$ with the CHC solver described in the previous subsection. The stack fragments $\sfrags$ are implemented using monadic state over an immutable set, which is more efficient than an imperative set since we need to checkpoint the current $\sfrags$ in the visited set $\visited$ per the rules.  We fix $k$ to $2$ as that is the minimal $k$ for which stitching is effective.  We include an additional \texttt{letassert} syntax in our language that we can use to write assertions on result values, which Spacer can then verify.  More details of the implementation are found in \infull{Appendix \ref{sec_implement_appendix}}\inshort{Appendix B of the Supplementary Material}.

One optimization we can perform is not reflected in the rules.  In practice, a non-local variable lookup can get stuck (meaning it will not be included in the result set) due to the inaccuracy of stack stitching in the all-paths analysis.  For example, $\texttt{x}^1$ is non-sensical to look up in \ruleref{Var Non-Local} for the case that the stack stitched has a top stack frame which is never invoking a function where \texttt{x} is a local variable.  To keep the analysis from wasting time going down such spurious paths, we precompute for each function a set of variables that can be looked up at that point.  To give an example, consider $(\gtfun\ \texttt{x} \gtarrow (\gtfun\ \texttt{y} \gtarrow (\texttt{x}^1\ \texttt{y}^0))^{\{\reallab{$\texttt{x}^1$},\ \reallab{$\texttt{y}^0$}\}})^{\{\reallab{$\texttt{x}^0$}\}}$.  Specifically, the set of variables carried by the inner function, \{\ \reallab{$\texttt{x}^1$},\ \reallab{$\texttt{y}^0$}\ \}, says that if at any point this function is invoked in the top frame of the stitched stack, then the non-local variable being looked up must be $\texttt{x}^1$ for this to be a valid stack with which to continue the analysis.  If this test fails the path is provably spurious and is pruned.

\subsubsection{Results Produced by the Analysis}
\label{sec_analysis_results}

\begin{wrapfigure}{r}{.5\textwidth}
    \ 
  \scriptsize
 \begin{camlnonum}
         let id = fun self $\gtarrow$ fun n $\gtarrow$
           if n = 0 then 0 
           else 1 + self self (n - 1) 
         in letassert r = id id 10 in r >= 2 
 \end{camlnonum}
    \centering
    \includegraphics[width=.25\linewidth]{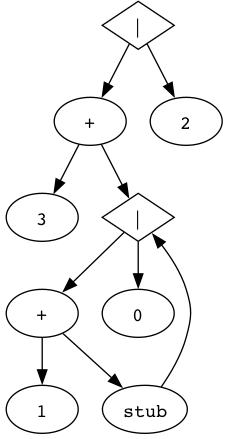}
  \caption{A recursive integer \texttt{id} example and its result}
  \label{fig_id}
\end{wrapfigure}

Here we summarize some of the results produced by the current implementation. 

Figure \ref{fig_id} gives the source code of a simple recursive identity function and a graph of the result automatically produced by the analysis showing the recurrences: both the \texttt{stub} from-node and the parent to-node on the back edge have the same $\la\cycmark,\stk\ra$. The path conditions are elided from the result as they do not refine it in this case. Path sensitivity does refine the running of the analysis; for example, 0 and 1 are ruled out as results since before the analysis hits a cycle (\eg a previously visited state in $\visited$), Spacer can show that \texttt{10} and \texttt{10 - 1} are both non-zero and so the base case can be skipped in the \ruleref{Conditional} rule for the first two recursions.  Spacer can prove that the result here is greater than or equal to two based on the recurrence result extracted: the \camlil!letassert! in the source is fed into Spacer and is reported satisfiable.  Note that the \camlil!let! syntax is sugar that we inline in the frontend.

\begin{wrapfigure}{r}{.5\textwidth}
  \scriptsize
\begin{camlnonum}
         let id = (fun xx $\gtarrow$ xx) in
         let map = (fun f $\gtarrow$ fun l $\gtarrow$
           let lp = fun self $\gtarrow$ fun lst $\gtarrow$
             if not (hd in lst) then {}
             else ({ hd = id f (lst.hd); 
                     tl = self self (lst.tl) })
           in
           lp lp l) in
         map (id (fun b $\gtarrow$ 1 + b)) 
                ({ hd = 7; tl = { 
                   hd = 8; tl = { 
                   hd = 9; tl = {} } } })
\end{camlnonum}
   \centering
    \includegraphics[width=.5\linewidth]{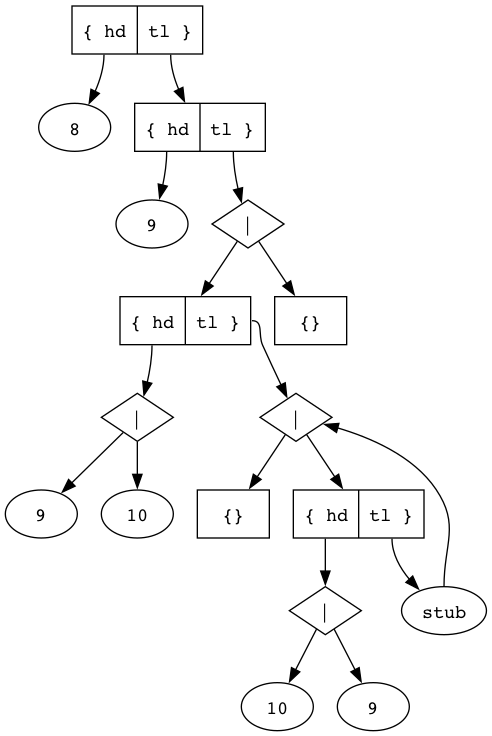}
  \caption{A \texttt{map} example and its result}
  \label{fig_map}
\end{wrapfigure}

The current proof-of-concept implementation quickly completes many examples but still needs refinement before being practically useful. For an idea of the performance, we compare it with the benchmarks run on DDPA \cite{DDPA} and P4F \cite{P4F}, two state-of-the-art higher-order program analyses.

Since neither P4F nor DDPA has any accuracy on integer value ranges, we also implemented a simplified version of the pure demand program analysis that analogously conflates all integer values, and also does not call out to Z3/Spacer to solve integer recurrences as does the full analysis described in the previous section.

Table \ref{tab_benchmarks} presents the results; we only include in the table the benchmarks that our current grammar supports and that appears in either of these two papers.  Of the 19 benchmarks listed in DDPA \cite[Figure 33]{DDPA}, there are three that fall outside of our language syntax, so there are 16 remaining we can compare with.  Of these 16, 10 also appear in P4F \cite[Figures 2 \& 3]{P4F}, and we were also able to run P4F on the other six not appearing directly in that paper.  To produce comparable data we downloaded the artifacts from the two above cited papers and re-ran all the benchmarks in those papers on our hardware.  We used the OCaml {\tt Core\_bench} library \cite{Core_bench} to run our own pure demand analysis benchmarks.  In the table, the "Full System" includes the recurrence solving, and the "Simplified" system has no recurrences and conflates all integers.

Inspecting the table, one can observe a striking bimodal aspect to the runtime performance: the pure demand system running times were generally either \emph{much} faster or \emph{much} slower than DDPA/P4F.  We believe this reflects how the stack stitching approach is a fundamentally different algorithm compared to either P4F or DDPA, both of which rely on a similar stack context model and so end up being relatively  similar in performance.

We now discuss some of the results in more detail.  The ack, tak, and cpstak benchmarks are all EXPTIME functions and we see a very strong exponential blowup for the pure demand analysis on exponential functions: neither version of the pure demand analysis terminates on any of these functions.  On all of the non-exponential benchmarks, our performance is either better or much better, the only exception being loop2-1, which is a complex higher-order doubly-nested loop (it is faster in the full system because the expressive flow-sensitivity there helps speed up the computation).

In terms of expressiveness of the analysis results, the pure demand analysis is either more accurate than the DDPA benchmarks (seven), as accurate as them (three), or less accurate than them (three) in terms of the benchmarks that terminate within five minutes; see the last column of Table \ref{tab_benchmarks} for the details.  For instance, our analysis derives $\gttrue$ from the blur benchmark, whereas DDPA derives $\{\gttrue,\gtfalse\}$.  From benchmarks like facehugger and map, our analysis derives integer recurrences while DDPA is only able to deduce that the result is an integer.  For the sat benchmarks DDPA can accurately solve satisfiability on the particular formulae of those benchmarks whereas our full system cannot. DDPA is at least as accurate as P4F in all cases so we elide a comparison with it.  Note that the DDPA benchmarks in the paper were run both with no context-sensitivity and with maximal context-sensitivity; we are only comparing with the maximally context-sensitive (\ie maximally accurate) version here. 

\begin{table}[h]
  \newcolumntype{C}[1]{>{\centering\let\newline\\\arraybackslash\hspace{0pt}}m{#1}}
  \newcolumntype{R}[1]{>{\raggedleft\let\newline\\\arraybackslash\hspace{0pt}}m{#1}}

  \begin{center}
    \begin{tabular}{ R{5em} || C{4.5em}| C{4.5em} || C{3.5em} | C{3.5em} || C{4.5em} } 
      Benchmark
      & \begin{tabular}{@{}c@{}} Pure \\ Demand \\ Full \\ System \end{tabular}
      & \begin{tabular}{@{}c@{}} Pure \\ Demand \\ Simplified \\ System \end{tabular}
      & DDPA & P4F
      & \begin{tabular}{@{}c@{}} Full \\ System \\ Accuracy \\ vs. DDPA \end{tabular}\\
      \hline
      blur       & 8.3     & 1.6      & 99.3    & 125.4   & greater  \\ 
      eta        & 0.1     & 0.1      & 20.4    & 61.4    & equal    \\ 
      facehugger & 54.7    & 1.1      & 89.2    & 109.1   & greater  \\ 
      kcfa-2     & 0.1     & 0.1      & 60.4    & 94.9    & equal    \\ 
      kcfa-3     & 0.2     & 0.2      & 90.2    & 250.1   & equal    \\ 
      mj09       & 4.2     & 0.1      & 45.3    & 72.8    & greater  \\ 
      loop2-1    & 1188    & 11106.4  & 198.8   & 121.1   & greater  \\ 
      map        & 21      & 1.5      & 220.0   & 142.8   & greater  \\ 
      primtest   & 18.1    & 39.3     & 1266.2  & 640.3   & greater  \\ 
      rsa        & 9.4     & 0.6      & 1139.3  & 3177.6  & greater  \\ 
      sat-1      & 117.6   & 117.2    & 728.6   & 236.4   & less     \\ 
      sat-2      & 50876.6 & 51811.5  & 71416.0 & 32367.0 & less     \\ 
      sat-3      & 5204.7  & 5215.1   & 47712.6 & 4627.6  & less     \\ 
      ack        & --      & --       & 287.5   & 112.6   & --       \\ 
      tak        & --      & --       & 1131.5  & 166.7   & --       \\ 
      cpstak     & --      & --       & 2738.5  & 166.4   & --       \\ 
    \end{tabular}
  \end{center}
  
  \caption{Benchmark runtimes across program analyses (ms). -- indicates timeout after five minutes.}
  \label{tab_benchmarks}
\end{table}

\subsubsection{Improving Performance}
\label{sec_impl_improv_perf}

The initial results reported above are very much initial results.  While we have made some effort to optimize data structures and simplify results when possible, there still are considerable performance improvements that could be made.

Caching in particular is very important for performance as pure demand semantics must do considerably more work looking up variables if there is no cache.
For the interpreter implementation described above, we used a mapping  $(\stk,\expr) \mapsto \reval$ as a cache and obtained interpreter running times similar to a standard implementation of an interpreter.  Caching the $\stk, \sfrags, \pathc, \visited \vdashea  \expr \steps \reval \globs \sfrags$ relation of the analysis is more challenging: in order to be sound, we can use a cached result only when the \emph{exact} same question is asked again, and this means a cache requires a mapping $(\stk, \sfrags, \pathc, \visited,\expr) \mapsto \reval$, greatly limiting the number of cache hits: we implemented such a cache and found the cache was hardly ever consulted.

What would be desirable would be a way to soundly merge some of these cache keys to increase the number of cache hits.  One simple experiment we tried was to completely remove the $\visited$ set from the keys, which has the property of shifting when a cycle gets pruned.  In our experiments, it was nearly always innocuous, and produced runtimes which were invariably much, much faster than the DDPA/P4F runtimes, but unfortunately it wasn't \emph{always} innocuous: on the cpstak benchmark (only) it produced an incorrect output.  We believe there should be provably sound ways to increase the number of cache hits, but it is a subtle question that we leave to future work.


\section{Related Work}
\label{sec_related}


\subsection{Related Operational Semantics}

The operational semantics here only uses the call stack and a source program point or a lexical level as its state.  We know of no such precedent, but the DDSE operational semantics \cite{DDSE} relies only on a call stack, a non-locals stack, and a control flow graph.

The theme of replaying previous computations is not completely new \cite{KoppelCapturing} -- in this paper the technique is used to implement delimited continuations with only exceptions and state.

For imperative languages supporting nested function definitions, \emph{access links} are a classic compiler implementation for non-local variable access \cite{Fischer05,MitchellCompilers}.  Access links are a chain of activation records linking to successively higher lexical levels.  Our approach is similar in how there is a back-chaining to find values, but access-link-based implementations still push data forward in stack frames so we are more demand-driven.  DDSE uses a similar approach.  

Optimal lambda-reduction \cite{Lamping:1989,Gonthier:1992} is a $\lambda$-reduction method that provably requires the fewest number of $\beta$-reductions.  Running our interpreter with a cache has the related property that applications are never repeated.  Whether any concrete optimality property can be proven about a pure demand semantics is an open question.

We only address pure functional features in this paper, and it is not obvious how it would be extended to incorporate effects.  The non-local variable lookup currently is something of an ``escape'' to a different context which bears some analogy to control effects; this hints that an algebraic effects basis may allow general side effects to be modeled.

\subsection{Related Program Analyses}
\label{sec_analysis_related}

We believe our approach to program analysis for functional languages is nearly unique in that it is not directly propagating any data; all state is captured in the call stack and a program point or a lexical level.  Functional language program analyses in the literature invariably push (abstract) values forward.  Dataflow analysis for imperative programs also may push data backwards, but again data is being propagated directly.

All program analyses for languages with higher-order recursive functions must somehow finitize the call stack.  The classic approach here is $k$CFA \cite{ShiversThesis}, which is an abstract interpretation \cite{Cousot77} that trims the stack to the most recent $k$ calls.  Many extensions and improvements to $k$CFA have been made over the years including \cite{AAM,P4F,CFA2,ResolvingkCFAParadox,MightAbstractInterpretersForFree,JagaWeeksUnified}. 
To the best of our knowledge, our stitching of past stack histories to approximate past stack states is a novel feature of our analysis.  It is particularly important for demand variable lookup. 

The higher-order program analysis closest to this work is DDPA \cite{DDPA} and its generalization Plume \cite{Plume}, which also uses an underlying form of operational semantics with no substitution or environment.  By contrast, DDPA-like analyses rely on more than just a call stack to model program state: a flow graph and an additional non-locals stack are needed.  Here we perform a more direct finitization of the (infinitary) operational semantics by simply finitizing the rules by stack pruning, in the spirit of Abstract Interpreters for Free \cite{MightAbstractInterpretersForFree}; DDPA translates analysis questions into equivalent PDA reachability questions.  DDPA does not perform stack stitching, but it would likely benefit from it since it also looks up variables on demand.  DDPA is polynomial in $k$, but needs a rapidly-growing $k$ to maintain precision.  DDPA does not extract recurrence results as we do here.

The standard numerical abstractions in the field of abstract interpretation \cite{abstract-interp-domains-tutorial} include, for example, numerical sign (positive, zero, or negative), constant-value, set-of-constants, interval, and congruence domains.  The recurrences our analysis infers generalize these domains as it can infer sign, constant, and range values. Relational domains in abstract interpretations additionally can relate multiple variables.  For example, two variables could only be known to be positive separately, but it could be known collectively that one is always larger than the other.  Relational domains should in principle be inferrable in an extension of our analysis.

Our analysis is built on a big-step formalization of operational semantics; this is another non-standard aspect as a small-step basis is usually used, but there are several exceptions \cite{Schmidt1998,BodinEtalSkeletalSemantics,DaraisAbstractingDI}.  Big step operational semantics have an advantage of calls and returns always aligning since the call and return point is the same node in the proof tree.

We show in Section \ref{sec_all_paths_core} how the program analysis algorithm can be declaratively specified via a single deterministic proof system.  We believe this approach is novel, and the greater degree of rigor of a fully proof-theoretic presentation will be potentially useful in  formal verification.  Many approaches define program analyses declaratively as Datalog programs \cite{WhaleyLamDatalogAnalysis}; this gives a more declarative specification then pseudocode, but a full specification then also requires defining the meaning of a Datalog interpreter.  It is also possible to mathematically specify a program analysis as a greatest fixed point of an operational semantics \cite{BodinEtalSkeletalSemantics}, but this doesn't constitute a computable specification for an implementation.

We know of no higher-order program analysis that extracts recurrence relations on numerical data as we do, but the more general theme of analyses aiming to extract stronger invariant properties is an active area of research \cite{StableRelsAI,GowJaggaHOShape,JagaShapeInf}. We currently do not infer inductive input-output relationships \cite{StableRelsAI} but that should be possible in principle by synchronizing input and output recurrences.



\section{Conclusions}
\label{sec_concl}

This paper addresses the question: what is the minimal amount of state information that must be maintained in the operational semantics of a higher-order functional language?  The answer here is it suffices to \emph{only} use a call stack consisting of source program call sites, plus a source program point or a lexical level; no variable environment, substitution, store, or any other structure is needed.  We prove this form of operational semantics equivalent to a standard presentation.

Then, to illustrate how this pure demand approach can open the door to new applications, we defined a program analysis as a direct finitization of this operational semantics.  The program analysis has a number of novel properties compared to standard program analyses for higher-order programs.  The analysis has been implemented and preliminary results are promising.

\section*{Data-Availability Statement}
The implementation of the interpreter of Section \ref{sec_opsem} and the analyses of Sections \ref{sec_analysis} and \ref{sec_implementation} are publicly available on Zenodo \cite{SmithZhangArtifact24a}, Software Heritage \cite{SmithZhangArtifact24b}, and in an evolving open-source GitHub repository \cite{SmithZhangArtifact24c}.

\section*{Acknowledgements}
We thank the OOPSLA reviewers for numerous helpful remarks, and for providing the suggestion to use DeBruijn indices in place of the less elegant program label approach we were previously using.  We also thank Zachary Palmer, Earl Wu, and Shiwei Weng for numerous helpful discussions and comments.

\bibliographystyle{ACM-Reference-Format}
\bibliography{main}

\ifdefined\gh
\else
  \newcommand\gh[1]{\href{#1}{\color{black}\tiny\faLink}}
\fi

\infull{
  \appendix
  \newpage
  \section{Proofs}
\label{sec_proofs}
This appendix contains proofs of all lemmas and theorems from the paper.

\subsection{Operational Semantics Proofs}

\variablelookupcanonical*

\begin{proof}
    Since \ruleref{Var Non-Local} in turn always requires another variable lookup and no rules other than \ruleref{Var Local} and \ruleref{Var Non-Local} apply to variable expressions, it must be the case that all variable lookups form a chain in the proof tree consisting of $n-1$ \ruleref{Var Non-Local} applications followed by a \ruleref{Var Local} that ends the sequence.  Figure \ref{fig_variable_lookup_canonical} is then nothing more than a description of this general proof tree.
\end{proof}

\lemvc*

\begin{proof}
By Lemma \ref{lem_variable_lookup_canonical}, we know that any variable lookup sub-proof in $\vdash$ is of the form of Figure \ref{fig_variable_lookup_canonical}.  Now, observe that if all of the non-$\ev$ lookup preconditions of this subtree are assembled in a single rule, we have the \ruleref{Var} rule of $\vdashc$ by inspection, showing that the two systems are equivalent in power.
\end{proof}

\lemccc*

\begin{proof}
Before starting the proof proper, we need to define some notations.

We will use the convenience notation $\nofl(\eval^{\stk,\cachef,\cachel}) = \eval^\stk$.

We need a stronger invariant of ``what is in the cache is provable'' than simply asserting that property for the top-level cache; we also need the property that any cache inside any result also is provable.  For this stronger invariant, we need to define a ``deep containment'' relation $\in\in$ reflecting that a cached item may be in one of the nested caches.  Define $\la \stk_0,\cachef_0,\cachel_0,\expr_0\ra \mapsto \eval_0^{\stk_0',\cachef_0',\cachel_0'} \in\in \cachef$ to mean either $\la \stk_0,\cachef_0,\cachel_0,\expr_0\ra \mapsto \eval_0^{\stk_0',\cachef_0',\cachel_0'}  \in \cachef$, or for some $\stk_1, \cachef_1, \cachel_1, \expr_1$ with $\la \stk_1,\cachef_1,\cachel_1,\expr_1\ra \mapsto \eval_1^{\stk_1',\cachef_1',\cachel_1'} \in \cachef$ that we inductively have $\la \stk_0,\cachef_0,\cachel_0,\expr_0 \ra \mapsto \eval_0^{\stk_0',\cachef_0',\cachel_0'} \in\in \cachef_1 \cup \cachef'_1$.  The definition for $\in \in \cachel$ is  analogous.

  We show each direction in turn. 

\noindent \textsc{$\Rightarrow$ Direction:}\\
We generalize the assertion to show that for all $\expr$, $\stk$, $\cachef$, and $\cachel_0$, if
\begin{enumerate}
  \item $\stk \vdashc \expr \steps \eval^{\stk'}$,
  \item if $\stk = (\expr_0\ \expr_0') \lC \stk_0$ (also implying $\stk$ is non-empty) then
  \begin{enumerate}
    \item $\cachef$ is the list of all lexically-enclosing function lookups for $\expr$: 
  $$\cachef = \lC_{0 \leq i \leq m}[\la \stk_i,\cachef_i',\cachel_i',\expr_i\ra \mapsto (\gtfun\ \ev_{i} \gtarrow \expr_{i}'' )^{\stk_{i+1}',\cachef_{i+1},\cachel_{i+1}}]$$
  where for all $i \leq m$, $\stk_{i+1}' = ((\expr_{i+1}\ \expr'_{i+1}) \lC \stk_{i+1})$, and we also have $\stk_i\vdashc \expr_i \steps (\gtfun\ \ev_{i} \gtarrow \expr_{i}'' )^{\stk_{i+1}',\cachef_{i+1},\cachel_{i+1}}$,
  \item Each $\cachel_i$ caches the lookup of the top argument on the stack of the above lexical chain: for all $0 \leq i \leq m$,
  $$\cachel_i = [\la \stk_i,\cachef_i',\cachel_i',\expr_i'\ra \mapsto (\reval_i,\ev_i)] \text{ and } \stk_i \vdashc \expr_i' \steps (\reval_i,\ev_i),$$
  \end{enumerate}
\end{enumerate}
then there exist $\cachef$, $\cachel_0$ such that
\begin{enumerate}\renewcommand{\theenumi}{\roman{enumi}}
  \item $\stk, \cachef, \cachel_0 \vdashcc \expr \steps \eval^{\stk',\cachef',\cachel'_0}$,
 \item if $\stk' = (\expr_0\ \expr_0') \lC \stk_0'$ (also implying $\stk'$ is non-empty) then
 \begin{enumerate}
   \item 
 $$\cachef' = \lC_{0 \leq i \leq m}[\la \stk_i',\cachef_i'',\cachel_i'',\expr_i\ra \mapsto (\gtfun\ \ev_{i} \gtarrow \expr_{i}'' )^{\stk_{i+1}',\cachef_{i+1}',\cachel_{i+1}'}]$$
 where for all $i \leq m$, $\stk_{i+1}'' = ((\expr_{i+1}\ \expr'_{i+1}) \lC \stk_{i+1}')$, and additionally $\stk_i'\vdashc\expr_i \steps (\gtfun\ \ev_{i} \gtarrow \expr_{i}'' )^{\stk_{i+1}',\cachef_{i+1}',\cachel_{i+1}'}$,
 \item Each $\cachel_i'$ caches the lookup of the top argument on the stack of the above lexical chain: for all $0 \leq i \leq m$,
 $$\cachel_i' = [\la \stk_i',\cachef_i'',\cachel_i'',\expr_i'\ra \mapsto (\reval_i,\ev_i)] \text{ and } \stk_i'\vdashc \expr_i' \steps (\reval_i,\ev_i).$$
 \end{enumerate}\end{enumerate}

Note that this generalized statement implies the result because picking the caches to be empty will make (2) vacuously true.

Proceed by induction on the height of the assumption (1) proof tree to show that (1)-(2) implies (i)-(ii).

For the base case of the \ruleref{Value} rule, the conclusion is direct: the assumption (1) in this case is $\stk \vdashc \eval \steps \eval^{\stk}$ and (ii) is immediate from (2), and (i), $\stk, \cachef, \cachel_0 \vdashcc \eval \steps \eval^{\stk,\cachef, \cachel_0}$, is immediate from the $\vdashcc$ \ruleref{Value} rule.

For the induction step, assume the goal for smaller proof trees, and assume arbitrary $\cachef$, $\cachel_0$ with (1)-(2) holding, and (1) being $\stk \vdashc \expr \steps \reval$. 

For the \ruleref{Application} case, we must have $\expr = (\expr_1\ \expr_2)$ so we have $\stk \vdashc (\expr_1\ \expr_2) \steps \reval$.  We can directly apply the induction hypothesis for the function and argument lookup subgoals in the rule, choosing $\cachef$ and $\cachel_0$ there to be the assumed arbitrary ones above.  Since these caches are the same, we directly inherit the requirement (2) for the induction hypotheses of these sub-derivations from the assumed (2) for the conclusion, and so obtain from (i) for the first hypothesis that $\stk, \cachef,\cachel_0 \vdashcc \expr_1 \steps (\gtfun\ \ev \gtarrow \expr')^{\stk_1,\cachef',\cachel_0'}$ for some $\expr',\stk_1,\cachef',\cachel_0'$, and similarly obtain $\stk, \cachef,\cachel_0 \vdashcc \expr_2 \steps \eval_2^{\stk'',\cachef'',\cachel_0''}$ for some $\eval_2,\stk'',\cachef'',\cachel_0''$ for the second hypothesis. In order to apply the inductive hypothesis to the established 
$$((\expr_1\ \expr_2) \lC \stk) \vdashc \expr' \steps \reval$$
rule precondition, we need to construct caches such that (2) holds for this hypothesis.  Using the $\vdashcc$ \ruleref{Application} rule as a guidepost, pick the cache for this precondition to be $(\la \stk,\cachef,\cachel_0,\expr_1\ra \mapsto (\gtfun\ \ev \gtarrow \expr')^{\stk_1,\cachef',\cachel'} \lC\cachef')$, and pick the $\cachel_0$ in that assumption to be $[\la \stk,\cachef,\cachel_0,\expr_2\ra \mapsto (\reval_2,\ev)]$.  Precondition (2) then follows because by inspection of (2)'s definition we can see we have just added one more frame to the front of the $\cachef'$ chain that preserves (2).  So, with all preconditions (2) satisfied for each rule we have both (i) and (ii) for all three preconditions.  In particular for the third precondition we now have 
$((\expr_1\ \expr_2) \lC \stk)$ and $(\la \stk,\cachef,\cachel_0,\expr_1\ra \mapsto (\gtfun\ \ev \gtarrow \expr')^{\stk_1,\cachef',\cachel'} \lC\cachef'),[\la \stk,\cachef,\cachel_0,\expr_2\ra \mapsto (\reval_2,\ev)] \vdashcc \expr' \steps \reval$. Putting together the three $\vdashcc$ proofs established above, we may apply the $\vdashcc$ \ruleref{Application} rule to obtain the desired $\stk,\cachef,\cachel_0 \vdashcc \expr \steps \eval^{\stk',\cachef'',\cachel_0''}$, and (ii) follows directly as they are the same result caches as those occurring on the third assumption conclusion.

For the \ruleref{Var} case, assumption (1) is of the form $((\expr_0\ \expr_0')\lC \stk_0)\vdashc \ev_n \steps \reval$.  By assumption (2), we have that for all $0 \leq i \leq m$, $\la \stk_i,\cachef_i',\cachel_i', \expr_i \ra \mapsto (\gtfun\ \ev_{i} \gtarrow \expr_{i}'' )^{\stk_{i+1}',\cachef_{i+1},\cachel_{i+1}}$, where for all $0 \leq i < n$, $\stk_{i+1}' = ((\expr_{i+1}\ \expr'_{i+1}) \lC \stk_{i+1})$, and in addition for all $0 \leq i \leq m$, $\la \stk_i,\cachef_i',\cachel_i', \expr_i'\ra \mapsto (\reval_i,\ev_i)$.  These subsume the preconditions of the $\vdashcc$ \ruleref{Var} rule, noting that we want to pick the $n$th clause of the former mapping, where $\reval = \reval_n$ by Lemma \ref{lem_deterministic}: there is only one possible proof tree so these variables must be identical. So we have established (i), namely $((\expr_0\ \expr_0')\lC \stk_0), \cachef,\cachel_0 \vdashcc \ev_n \steps \reval$.  For requirement (ii), take the established assumptions $\forall 0 \leq i \leq n.\ \stk_i \vdashc \expr_i \steps (\gtfun\ \ev_{i} \gtarrow \expr_{i}'' )^{\stk_{i+1}'}$ and $\forall 0 \leq i < n.\ \stk_{i+1}' = ((\expr_{i+1}\ \expr'_{i+1}) \lC \stk_{i+1})$ plus $\stk_n \vdashc \expr'_n \steps \reval$ of the $\vdashc$ \ruleref{Var} rule.  We can establish that (2) holds for each of these preconditions by induction on $n$.  For the base case, the caches are the ones in the rule itself so it is trivial.  The inductive case is direct. So, in particular we have (ii) for the assumption $\stk_n \vdashc \expr'_n \steps \reval$, and since the $\vdashcc$ \ruleref{Var} rule inherits these caches in its result, this establishes (ii) for this case.  The $\Rightarrow$ direction is now complete.


\noindent \textsc{$\Leftarrow$ Direction:}\\
Assume for $\vdashcc$ and show for $\vdashc$.  We generalize the assertion to show that if 
\begin{enumerate}
  \item $\stk, \cachef, \cachel \vdashcc \expr \steps \eval^{\stk',\cachef',\cachel'}$,
  \item For each $\la \stk_0,\cachef_0,\cachel_0,\expr_0\ra \mapsto \reval_0 \in\in \cachef$, we have $\stk_0,\cachef_0,\cachel_0 \vdashcc \expr_0 \steps \reval_0$ and $\stk_0 \vdashc \expr_0 \steps \nofl(\reval_0)$,
  \item For each $\la \stk_0,\cachef_0,\cachel_0,\expr_0\ra \mapsto (\reval_0,\ev_0) \in\in \cachel$, we have $\stk_0,\cachef_0,\cachel_0 \vdashcc \expr_0 \steps \reval_0$ and $\stk_0 \vdashc \expr_0 \steps \nofl(\reval_0)$,
\end{enumerate}
then 
\begin{enumerate}\renewcommand{\theenumi}{\roman{enumi}}
  \item $\stk \vdashc \expr \steps \eval^{\stk'}$,
  \item For each $\la \stk_0,\cachef_0,\cachel_0,\expr_0\ra \mapsto \reval_0 \in\in \cachef'$, we have $\stk_0,\cachef_0,\cachel_0 \vdashcc \expr_0 \steps \reval_0$ and $\stk_0 \vdashc \expr_0 \steps \nofl(\reval_0)$,
  \item For each $\la \stk_0,\cachef_0,\cachel_0,\expr_0\ra \mapsto (\reval_0,\ev_0) \in\in \cachel'$, we have $\stk_0,\cachef_0,\cachel_0 \vdashcc \expr_0 \steps \reval_0$ and $\stk_0 \vdashc \expr_0 \steps \nofl(\reval_0)$.
\end{enumerate}

Note that this generalized statement implies the result because the caches are empty there and so the strengthened assumptions (2)-(3) are vacuously true.

Proceed by induction on the height of the sub-proof of $\gexpr$ to show that (1)-(3) implies (i)-(iii).

For the base case of the \ruleref{Value} rule, the conclusion is direct: the assumption (1) in this case is $\stk, \cachef, \cachel \vdashcc \eval \steps \eval^{\stk,\cachef,\cachel}$ and (ii) is immediate from (2), (iii) from (3), and (i), $\stk \vdashc \eval \steps \eval^{\stk}$, is immediate from the $\vdashc$ \ruleref{Value} rule.

For the induction step, assume the goal for smaller proof trees, and assume (1)-(3) with (1) being $\stk, \cachef, \cachel \vdashcc \expr \steps \reval$.  

For the \ruleref{Application} case, we must have $\expr = (\expr_1\ \expr_2)$ so we have $\stk, \cachef, \cachel \vdashcc (\expr_1\ \expr_2) \steps \reval$.  We can directly apply the induction hypothesis for the function and argument lookup subgoals in the rule: since the input caches do not change, we directly inherit the requirements (2) and (3) for the induction hypotheses of these sub-derivations and obtain from (i) both $\stk \vdashc \expr_1 \steps (\gtfun\ \ev \gtarrow \expr')^{\stk_1}$ for some $\expr'$ and $\stk \vdashc \expr_2 \steps \eval_2^{\stk'}$.  In order to apply the inductive hypothesis to the established 
$$((\expr_1\ \expr_2) \lC \stk), (\la \stk,\cachef,\cachel_0,\expr_1\ra \mapsto (\gtfun\ \ev \gtarrow \expr')^{\stk_1,\cachef',\cachel'} \lC\cachef'), [\la \stk,\cachef,\cachel_0,\expr_2\ra \mapsto (\reval_2,\ev)] \vdashcc \expr' \steps \reval$$
rule precondition, we need to show that (2) and (3) hold for the caches in this assumption. For (2), this fact is direct as we established the preconditions (1)-(3) for the first assumption in the rule and so can conclude (ii) for the mappings $\cachef'$ in this assumption. The new mapping added in the precondition, $\la \stk,\cachef,\cachel_0,\expr_1\ra \mapsto (\gtfun\ \ev \gtarrow \expr')^{\stk_1,\cachef',\cachel'}$, is exactly the caching of the first assumption of the $\vdashcc$ rule and so by (ii) for that assumption, we have that this mapping and all component mappings are provable.  So, every mapping in the function cache $(\la \stk,\cachef,\cachel_0,\expr_1\ra \mapsto (\gtfun\ \ev \gtarrow \expr')^{\stk_1,\cachef',\cachel'} \lC\cachef')$ is provable in both systems.  For precondition (3), the cache $[\la \stk,\cachef,\cachel_0,\expr_2\ra \mapsto (\reval_2,\ev)]$ of the rule precondition is exactly the second precondition of the $\vdashcc$ rule.  So, given that (ii) and (iii) hold for this precondition as established above, we can conclude 
$((\expr_1\ \expr_2) \lC \stk) \vdashc \expr' \steps \nofl(\reval)$ as well as (ii) and (iii) for the result caches.  Putting the three $\vdashc$ proofs established above, we may apply the $\vdashc$ \ruleref{Application} rule to obtain the desired $\stk \vdashc \expr \steps \nofl(\reval)$, and (ii) and (iii) follow directly as they are the same result caches as those occurring on the third assumption.

For the \ruleref{Var} case, assumption (1) is of the form $((\expr_0\ \expr_0')\lC \stk_0), \cachef, \cachel \vdashcc \ev_n \steps \reval$.  By the assumptions of this rule and the fact that all cached mappings are $\vdashcc$-provable by assumptions (2) and (3), we have that for all $0 \leq i \leq m$, $\stk_i,\cachef_i',\cachel_i' \vdashcc \expr_i \steps (\gtfun\ \ev_{i} \gtarrow \expr_{i}'' )^{\stk_{i+1}',\cachef_{i+1},\cachel_{i+1}}$, where for all $0 \leq i < n$, $\stk_{i+1}' = ((\expr_{i+1}\ \expr'_{i+1}) \lC \stk_{i+1})$, and in addition $\stk_n,\cachef_n',\cachel_n' \vdashcc \expr_n' \steps \reval$.  And, (2) and (3) also guarantee the analogous facts hold in the $\vdashc$ system and so we also have that for all $0 \leq i \leq m$, $\stk_i, \vdashc \expr_i \steps (\gtfun\ \ev_{i} \gtarrow \expr_{i}'' )^{\stk_{i+1}'}$, where for all $0 \leq i < n$, $\stk_{i+1}' = ((\expr_{i+1}\ \expr'_{i+1}) \lC \stk_{i+1})$, and in addition $\stk_n \vdashc \expr_n' \steps \nofl(\reval)$.  These are exactly the preconditions of the $\vdashc$ \ruleref{Var} rule and so we have established (i), namely $((\expr_0\ \expr_0')\lC \stk_0) \vdashc \ev_n \steps \nofl(\reval)$.  For requirements (ii) and (iii), the result caches here are those on $\reval$ from $\stk_n,\cachef_n',\cachel_n' \vdashcc \expr_n' \steps \reval$, and by the meaning of $\in\in$ on cached values, we also have that (ii) and (iii) hold for $\reval$'s caches.

\end{proof}

\lemcce*

\begin{proof}
  We will use the following shorthand notation in this proof: \\
   $\cachetoenvt((\gtfun\ \ev \gtarrow \expr)^{\stk,\cachef,\cachel}) = (\gtfun\ \ev \gtarrow \expr)^{\extractcl((\gtfun\ \ev \gtarrow \expr)^{\stk,\cachef,\cachel})}$.
  
  We show each direction in turn.

   \noindent \textsc{$\Rightarrow$ Direction:}\\
   We generalize the assertion to show that if
   \begin{enumerate}
     \item $\stk, \cachef, \cachel \vdashcc \expr \steps \eval^{\stk',\cachef',\cachel'}$,
     \item $\extractenvt(\stk,\cachef,\cachel) = \envt$ for some $\envt$
   \end{enumerate}
   then 
   \begin{enumerate}\renewcommand{\theenumi}{\roman{enumi}}
     \item $\envt \vdashe \expr \steps \eval^{\envt'}$ for some $\envt'$,
     \item $\extractenvt(\stk',\cachef',\cachel') = \envt'$.
   \end{enumerate}

   Note that this generalized statement implies the result because picking the caches to be empty will extract an empty environment.
   
   Proceed by induction on the height of the assumption (1) proof tree to show that (1)-(2) implies (i)-(ii).
   
   For the base case of the \ruleref{Value} rule, the conclusion is direct: the assumption (1) in this case is $\stk, \cachef, \cachel \vdashcc \eval \steps \eval^{\stk,\cachef,\cachel}$ and so (ii) is immediate from (2) picking $\envt' = \envt$ and (i), $\stk, \envt \vdashe \eval \steps \eval^{\envt}$, is immediate from the $\vdashe$ \ruleref{Value} rule.
   
   For the induction step, assume the goal for smaller proof trees, and assume (1)-(2) with (1) being $\stk, \cachef, \cachel \vdashcc \expr \steps \reval$.  
   
   For the \ruleref{Application} case, $\expr = (\expr_1\ \expr_2)$ so we have $\stk, \cachef, \cachel \vdashcc (\expr_1\ \expr_2) \steps \reval$.  We can directly apply the induction hypothesis for the function and argument lookup subgoals in the rule: since the input caches do not change we directly inherit the requirement (2) for the induction hypotheses of these sub-derivations and obtain from (i) both $\envt \vdashe \expr_1 \steps (\gtfun\ \ev \gtarrow \expr')^{\envt_1}$ for some $\expr'$ and $\envt \vdashe \expr_2 \steps \eval_2^{\envt'}$.  In order to apply the inductive hypothesis to the established 
   $$((\expr_1\ \expr_2) \lC \stk), (\la \stk,\cachef,\cachel_0,\expr_1\ra \mapsto (\gtfun\ \ev \gtarrow \expr')^{\stk_1,\cachef',\cachel'} \lC\cachef'), [\la \stk,\cachef,\cachel_0,\expr_2\ra \mapsto (\reval_2,\ev)] \vdashcc \expr' \steps \reval$$
   rule precondition, we need to show that (2) holds for the caches in this assumption. 
   By (ii) of the first assumption, we have that $\extractenvt(\stk_1,\cachef',\cachel') = \envt_1$ for some $\envt_1$.  From this and the definition of $\extractenvt$, we can obtain $$
   \begin{array}{c}
     \extractenvt(\stk_1,((\expr_1\ \expr_2) \lC \stk), (\la \stk,\cachef,\cachel_0,\expr_1\ra \mapsto (\gtfun\ \ev \gtarrow \expr')^{\stk_1,\cachef',\cachel'} \lC\cachef'),[\la \stk,\cachef,\cachel_0,\expr_2\ra \mapsto (\reval_2,\ev)])\\ = \\
     \ [\ev \mapsto \cachetoenvt(\reval_2)]\lC \envt_1
   \end{array}$$
   This establishes precondition (2) for the third hypothesis, so we may conclude (i) and (ii) for that hypothesis, namely
   $[\ev \mapsto \cachetoenvt(\reval_2)]\lC \envt_1 \vdashe \expr' \steps \cachetoenvt(\reval)$ for (i) and then (ii) is directly from the definitions.  Putting the three $\vdashe$ proofs established above we may apply the $\vdashe$ \ruleref{Application} rule to obtain the desired $\envt \vdashe \expr \steps \cachetoenvt(\reval)$, and (ii) follows directly as they are the same result caches as those occurring on the third assumption conclusion (ii).
   
   For the \ruleref{Var} case, assumption (1) is of the form $((\expr_0\ \expr_0')\lC \stk_0), \cachef, \cachel \vdashcc \ev_n \steps \reval$.  By assumption (2) we have $\extractenvt((\expr_0\ \expr_0')\lC \stk_0,\cachef,\cachel) = \envt$ for some $\envt$.  Further, from the definition of $\extractenvt$ it can be seen that $\ev_n \mapsto \cachetoenvt(\reval) \in \envt$, so $\envt \vdashe \ev_n \steps \cachetoenvt(\reval)$ by the $\vdashe$ \ruleref{Var} rule, establishing (i) for this case. For (ii), expanding definition it amounts to requiring $\extractcl(\reval) = \extractcl(\reval)$, which follows by reflexivity.

   
   \noindent \textsc{$\Leftarrow$ Direction:}\\
   Assume for $\vdashe$ and show for $\vdashcc$.  
   We generalize the assertion to show that for all $\stk$, $\cachef$, and $\cachel$, if
   \begin{enumerate}
     \item $\envt \vdashe \expr \steps \eval^{\envt'}$ for some $\envt'$,
     \item $\extractenvt(\stk,\cachef,\cachel) = \envt$
   \end{enumerate}
   then 
   \begin{enumerate}\renewcommand{\theenumi}{\roman{enumi}}
     \item $\stk, \cachef, \cachel \vdashcc \expr \steps \eval^{\stk',\cachef',\cachel'}$,
     \item $\extractenvt(\stk',\cachef',\cachel') = \envt'$.
   \end{enumerate}

   Note that this generalized statement implies the result because the stack, environment, and caches can be picked to be empty and so the strengthened assumption (2) is trivially true.
   
   Proceed by induction on the height of the sub-proof of $\gexpr$ to show that (1)-(2) implies (i)-(ii).
     
   For the base case of the \ruleref{Value} rule the conclusion is direct: the assumption (1) in this case is $\envt \vdashe \eval \steps \eval^{\envt}$ and (i), $\stk, \cachef, \cachel \vdashcc \eval \steps \eval^{\stk,\cachef, \cachel}$, is immediate from the $\vdashcc$ \ruleref{Value} rule.  (ii) is then immediate from (2). 
   
   For the induction step, assume the goal for smaller proof trees, and assume arbitrary $\stk, \cachef$, $\cachel$ with (1)-(2) holding, and (1) being $\envt \vdashe \expr \steps \reval$.  
   
   For the \ruleref{Application} case, we must have $\expr = (\expr_1\ \expr_2)$ so we have $\envt \vdashe (\expr_1\ \expr_2) \steps \reval$.  We can directly apply the induction hypothesis for the function and argument lookup subgoals in the rule, choosing $\stk$ $\cachef$ and $\cachel$ there to be the assumed arbitrary ones above.  Since these are the same we directly inherit the requirement (2) for the induction hypotheses of these sub-derivations from the assumed (2) for the conclusion, and so obtain from (i) for the first hypothesis that $\stk, \cachef,\cachel \vdashcc \expr_1 \steps (\gtfun\ \ev \gtarrow \expr')^{\stk_1,\cachef',\cachel'}$ for some $\expr',\stk_1,\cachef',\cachel'$, and similarly obtain $\stk, \cachef,\cachel \vdashcc \expr_2 \steps \eval_2^{\stk'',\cachef'',\cachel''}$ for some $\eval_2, \stk'',\cachef'',\cachel''$ for the second hypothesis. In order to apply the inductive hypothesis to the established 
   $$[\ev \mapsto \reval_2] \lC \envt_1 \vdashe \expr' \steps \reval$$
   rule precondition, we need to show (2) holds for this environment.  Using the $\vdashcc$ application rule as a guidepost, pick the $\stk$ for this precondition to be $\stk_1,((\expr_1\ \expr_2) \lC \stk)$, pick $\cachef$ to be $(\la \stk,\cachef,\cachel_0,\expr_1\ra \mapsto (\gtfun\ \ev \gtarrow \expr')^{\stk_1,\cachef',\cachel'} \lC\cachef')$, and pick the $\cachel$ to be $[\la \stk,\cachef,\cachel_0,\expr_2\ra \mapsto (\reval_2',\ev)]$.  We aim to show
   $$
   \begin{array}{c}
     \extractenvt(\stk_1,((\expr_1\ \expr_2) \lC \stk), (\la \stk,\cachef,\cachel,\expr_1\ra \mapsto (\gtfun\ \ev \gtarrow \expr')^{\stk_1,\cachef',\cachel'} \lC\cachef'),[\la \stk,\cachef,\cachel_0,\expr_2\ra \mapsto (\reval_2',\ev)])\\ = \\
     \ [\ev \mapsto \reval_2]\lC \envt_1
   \end{array}$$
   where $\cachetoenvt(\reval_2')=\reval_2$.
 
  Precondition (2) follows because by inspection of (2)'s definition we can see we have just added one more frame to the front of the $\cachef'$ chain which preserves (2).   So, with all preconditions (2) satisfied for each rule we have both (i) and (ii) for all preconditions.  In particular, for the third precondition we now have 
   $((\expr_1\ \expr_2) \lC \stk), (\la \stk,\cachef,\cachel_0,\expr_1\ra \mapsto (\gtfun\ \ev \gtarrow \expr')^{\stk_1,\cachef',\cachel'} \lC\cachef'),[\la \stk,\cachef,\cachel_0,\expr_2\ra \mapsto (\reval_2',\ev)] \vdashcc \expr' \steps \reval$.  Putting the three $\vdashcc$ proofs established above, we may apply the $\vdashcc$ \ruleref{Application} rule to obtain the desired $\stk,\cachef,\cachel_0 \vdashcc \expr \steps \eval^{\stk',\cachef'',\cachel_0''}$, and (ii) follows directly as they are the same result caches as those occurring on the third assumption conclusion.
   
   For the \ruleref{Var} case, assumption (1) is of the form $\envt \vdashe \ev_n \steps \reval$.  By assumption (2), we have $\extractenvt(\stk,\cachef,\cachel) = \envt$.  So, by expanding this definition it means we have for all $0 \leq i \leq m$ that $\la \stk_i,\cachef_i',\cachel_i', \expr_i \ra \mapsto (\gtfun\ \ev_{i} \gtarrow \expr_{i}'' )^{\stk_{i+1}',\cachef_{i+1},\cachel_{i+1}}$, where for all $0 \leq i < n$, $\stk_{i+1}' = ((\expr_{i+1}\ \expr'_{i+1}) \lC \stk_{i+1})$, and in addition for all $0 \leq i \leq m$, $\la \stk_i,\cachef_i',\cachel_i', \expr_i'\ra \mapsto (\reval_i,\ev_i)$.  These subsume the preconditions of the $\vdashcc$ \ruleref{Var} rule, noting that we want to pick the $n$th clause of the former mapping, where $\reval = \cachetoenvt(\reval_n)$ and so we have established (i), namely $((\expr_0\ \expr_0')\lC \stk_0), \cachef,\cachel \vdashcc \ev_n \steps \reval_n$.  Requirement (ii) is then direct from the definition of $\extractenvt$.
 \end{proof}
 
 \thmve*
 \begin{proof}
  By chaining Lemmas \ref{lem_v_c}, \ref{lem_c-cc}, and \ref{lem_cc-e}.
\end{proof}
  
\subsection{Analysis Proofs}
\coreanaldec*

\begin{proof}
  First, the set of possible $\stk$ are at most $k$ in length by inspection of the rules, and there are finitely many possible stack frames because the frames are invariably call sites in the source program.  And, results $\expr$ in the proofs are invariably subterms of the original $\gexpr$ (except for the case of relabeled variables, but those two sets are finite), so they are also finite.  $\sfrags$ is finite since it is a set of elements drawn from a finite domain.  So, $\stk, \sfrags \vdasha \expr \steps \reval$ has only finitely many possible quadruples to check and so all possible proofs can be enumerated and checked in finite time.
\end{proof}

\analysissoundness*

\begin{proof}[Proof Sketch]
  Before proceeding with the proof proper we need to define some notations.

  We use the notation $\stk_0 \ll_k \stk$ to indicate that $\stk_0$ occurs as a prefix of $\stk$ with $|\stk_0| = k$, or $\stk = \stk_0 \lC \stk'$ and $|\stk_0| < k$.  $\stk\lceil_k$ is the first $k$ elements of $\stk$, or is $\stk$ if $|\stk|$ is less than $k$.

  The theorem is proved by establishing the following more general assertion:

For arbitrary $\stk, e, \sfrags$ where for each $\stk_0 \ll_k \stk$, $\stk_0 \in \sfrags$, if  $\stk \vdash \expr \steps (\gtfun\ \ev \gtarrow \expr)^{\stk'}$ then $\stk\lceil_k, \sfrags \vdasha \expr \steps (\gtfun\ \ev \gtarrow \expr)^{\stk\lceil k} \globs \sfrags'$ for some $\sfrags'$. 

The theorem follows immediately from this statement picking $\stk = []$ and $\sfrags = \emptyset$.  We assume the antecedents and proceed by induction on the height of the $\vdasha$ proof.  Proceed by cases on the rule at the root of the proof tree.  Consider the \ruleref{Application} rule.  The first two subgoals of the abstract execution follows directly by induction from the analogous subgoals for the concrete execution.  For the third subgoal, since we have assumption that each $\stk_0 \ll_k \stk$, $\stk_0 \in \sfrags$, in the analogous concrete subgoal  $((e\ e')\lC \stk) \vdash \expr \steps \reval$ we are only adding one frame to the stack, and in the corresponding analysis rule we also add the same frame to the fragments set $(\sfrags_2 \cup ((e\ e') \lCk \stk))$, maintaining the invariant that all sub-stacks of the concrete execution $\stk$ are in this new fragment set.  So, this allows us to apply the induction hypothesis and obtain the second subgoal and thus the goal.

For the variable lookup rules, since $\sfrags$ includes all possible fragments of the actual run by assumption, there will be some stitching of those fragments that realizes this actual stack, resulting by induction in the analogous value returned.
\end{proof}

\aisaa*

We are currently calling the above only a conjecture because it is delicate to show that stubbings are invariably infinite cycles that will never be provable in the $\vdasha$ rules.  Here is an outline of how a proof could go.  The intuition is that the visited set $\visited$ is fully caching all previous "questions" asked on the current proof path, and the \ruleref{Stub} rules will only activate if the \emph{exact} same question is asked on the same path, indicating that we must be constructing an infinite proof tree which is not a proof in an inductive proof system.  The stack fragments $\sfrags$ need to be placed in the visited set to guarantee the question is exactly the same.  Besides the stubbing non-proof case, it is not hard to see that any individual run of the $\vdasha$ rules is realized as one deterministic slice of the nondeterministic set of runs of the $\vdashaa$ rules.

  \section{Implementation Code Overview}
\label{sec_implement_appendix}

This appendix gives a high-level code overview of the implementations described in the paper.  Links to our software artifact \cite{SmithZhangArtifact24b} are inserted wherever possible to help with contextualization.  We also maintain an open-source codebase \gh{https://github.com/JHU-PL-Lab/dde} \cite{SmithZhangArtifact24c} that will be evolved beyond what is presented in this paper.

\subsection{Interpreter}

In this section, we go over the architecture of the concrete interpreter described in Section \ref{sec_extended_language}.

\subsubsection*{Frontend}

We use OCaml's parser generator Menhir \cite{Menhir} as the frontend of our language.  We omit the details as its usage is standard.  

\subsubsection*{AST}

The frontend parses an input program into an abstract syntax tree (AST) represented by OCaml's algebraic data types \gh{https://archive.softwareheritage.org/swh:1:cnt:1efd49d53593fda7d2fe8fbe3593f9e9e92cf538;origin=https://github.com/JHU-PL-Lab/dde;visit=swh:1:snp:310bdd3c60f2c48c7916c12e7e70aa5eb815335e;anchor=swh:1:rev:943353fd2fa6d2d2bedc80b85d788aac1930febe;path=/interpreter/src/ast.ml;lines=25-47}.  The only difference from Figure \ref{fig_grammar_extended} is that applications are each additionally annotated with a label automatically generated by the parser that uniquely identifies it.  Call stacks, implemented as OCaml lists, contain these labels rather than the actual applications to help speed up cache lookups.  The reason is that the OCaml \texttt{Map}s and \texttt{Set}s we use leverage structural equality for lookups, and repeatedly parsing complex program expressions takes a significant toll on performance.  To help convert labels to applications, we maintain an auxhiliary \texttt{myexpr} \gh{https://archive.softwareheritage.org/swh:1:cnt:1efd49d53593fda7d2fe8fbe3593f9e9e92cf538;origin=https://github.com/JHU-PL-Lab/dde;visit=swh:1:snp:310bdd3c60f2c48c7916c12e7e70aa5eb815335e;anchor=swh:1:rev:943353fd2fa6d2d2bedc80b85d788aac1930febe;path=/interpreter/src/ast.ml;lines=182} hash table that maps labels to their corresponding expression.

There is no constructor for let bindings in the data type; they are instead treated as macros and are substituted away \gh{https://archive.softwareheritage.org/swh:1:cnt:1efd49d53593fda7d2fe8fbe3593f9e9e92cf538;origin=https://github.com/JHU-PL-Lab/dde;visit=swh:1:snp:310bdd3c60f2c48c7916c12e7e70aa5eb815335e;anchor=swh:1:rev:943353fd2fa6d2d2bedc80b85d788aac1930febe;path=/interpreter/src/ast.ml;lines=325-363} \emph{before} evaluation.  Another, effectively equivalent, approach is transforming them into applications \gh{https://archive.softwareheritage.org/swh:1:cnt:1efd49d53593fda7d2fe8fbe3593f9e9e92cf538;origin=https://github.com/JHU-PL-Lab/dde;visit=swh:1:snp:310bdd3c60f2c48c7916c12e7e70aa5eb815335e;anchor=swh:1:rev:943353fd2fa6d2d2bedc80b85d788aac1930febe;path=/interpreter/src/ast.ml;lines=367-392}.

\subsubsection*{Interpreter}

The operational semantics rules presented in Figure \ref{fig_steps_extended} are implemented as \texttt{eval\_aux} \gh{https://archive.softwareheritage.org/swh:1:cnt:fce359b15c9d99cf4bc0bcab5fa5f781cb29eecd;origin=https://github.com/JHU-PL-Lab/dde;visit=swh:1:snp:310bdd3c60f2c48c7916c12e7e70aa5eb815335e;anchor=swh:1:rev:943353fd2fa6d2d2bedc80b85d788aac1930febe;path=/interpreter/src/lib.ml;lines=79}.

To help cache repetitive computations as mentioned in Section \ref{sec_interp_efficiency}, \texttt{eval\_aux} threads an immutable map through a state monad \gh{https://archive.softwareheritage.org/swh:1:cnt:fce359b15c9d99cf4bc0bcab5fa5f781cb29eecd;origin=https://github.com/JHU-PL-Lab/dde;visit=swh:1:snp:310bdd3c60f2c48c7916c12e7e70aa5eb815335e;anchor=swh:1:rev:943353fd2fa6d2d2bedc80b85d788aac1930febe;path=/interpreter/src/lib.ml;lines=47-67}.  Its keys are tuples consisting of the call stack and the call site label or variable expression, and its values are evaluation results.  Monad operations like \texttt{bind} are flattened by using monadic let bindings provided by \texttt{ppx\_let} \cite{Ppx_let}.

Since only the operands of binary operations and record operations are evaluated and there is no subsitution of variables, we implement \texttt{eval\_result\_value} \gh{https://archive.softwareheritage.org/swh:1:cnt:fce359b15c9d99cf4bc0bcab5fa5f781cb29eecd;origin=https://github.com/JHU-PL-Lab/dde;visit=swh:1:snp:310bdd3c60f2c48c7916c12e7e70aa5eb815335e;anchor=swh:1:rev:943353fd2fa6d2d2bedc80b85d788aac1930febe;path=/interpreter/src/lib.ml;lines=340} to optionally simplify such results.  For example, \camlil!(fun y $\gtarrow$ fun x $\gtarrow$ y) 1! evaluates to \camlil!fun x $\gtarrow$ y!, plus implicitly the call stack in its closure.  Then, simplification involves substituting \texttt{1} for \texttt{y}, giving \camlil!fun x $\gtarrow$ 1!.  Notably, substitution is kicked off by invoking \texttt{eval\_aux} \gh{https://archive.softwareheritage.org/swh:1:cnt:fce359b15c9d99cf4bc0bcab5fa5f781cb29eecd;origin=https://github.com/JHU-PL-Lab/dde;visit=swh:1:snp:310bdd3c60f2c48c7916c12e7e70aa5eb815335e;anchor=swh:1:rev:943353fd2fa6d2d2bedc80b85d788aac1930febe;path=/interpreter/src/lib.ml;lines=278} on a version of the free variable \texttt{y} that can be looked up to a value starting from \camlil!fun x $\gtarrow$ y! with the call stack in its closure.

\subsection{Program Analysis (Full)}

The two program analyses implemented in our artifact both depend on the language frontend implemented as a part of the interpreter.  The simplified analysis (covered in the next section) is a stripped-down version of the full version that does not involve recurrence derivation and CHC solving, as mentioned in Section \ref{sec_analysis_results}. We first go over key parts of the full system.

\subsubsection*{Analysis} 

The operational semantics rules presented in Figure \ref{fig_steps_analysis_extended} are implemented as \texttt{analyze\_aux} \gh{https://archive.softwareheritage.org/swh:1:cnt:c7dcdc6f22fc27fecdedec1af11014125d36c918;origin=https://github.com/JHU-PL-Lab/dde;visit=swh:1:snp:310bdd3c60f2c48c7916c12e7e70aa5eb815335e;anchor=swh:1:rev:943353fd2fa6d2d2bedc80b85d788aac1930febe;path=/program_analysis/full/lib.ml;lines=164}.  It returns a stacked state-reader monad \gh{https://archive.softwareheritage.org/swh:1:cnt:2f8dda23d8881c12e679f1d980d00e3bced4e5e8;origin=https://github.com/JHU-PL-Lab/dde;visit=swh:1:snp:310bdd3c60f2c48c7916c12e7e70aa5eb815335e;anchor=swh:1:rev:943353fd2fa6d2d2bedc80b85d788aac1930febe;path=/program_analysis/full/utils.ml;lines=273-333} threading both mutable state (\eg the cache and the $\sfrags$ set) and immutable context (\eg the $\visited$ set).

Unlike the interpreter's simple cache key comprising of just the expression and call stack, caching an analysis run requires a tuple $(\stk, \sfrags, \pathc, \visited,\expr)$ as the cache key \gh{https://archive.softwareheritage.org/swh:1:cnt:2f8dda23d8881c12e679f1d980d00e3bced4e5e8;origin=https://github.com/JHU-PL-Lab/dde;visit=swh:1:snp:310bdd3c60f2c48c7916c12e7e70aa5eb815335e;anchor=swh:1:rev:943353fd2fa6d2d2bedc80b85d788aac1930febe;path=/program_analysis/full/utils.ml;lines=233-266}, as explained in Section \ref{sec_impl_improv_perf}.  To avoid costly set comparisons on each cache lookup, the $\sfrags$ and $\visited$ sets are represented by unique integer IDs generated each time a new set is created.  These optimizations contribute an noticeable improvement to the performance of our analysis.  We have also been exploring potential ways to soundly increase the cache hit rate and, by consequence, performance.

Analyzing applications and variables can lead to results containing ``lone'' stubs, \ie stubs that do not have a matching parent to form cycles with.  These are instances of potential divergence in the source program and can be safely pruned away.  Thus, \texttt{elim\_stub} \gh{https://archive.softwareheritage.org/swh:1:cnt:b8f2241c0d86e6b67873e25273ca4cbe6d5b01cd;origin=https://github.com/JHU-PL-Lab/dde;visit=swh:1:snp:310bdd3c60f2c48c7916c12e7e70aa5eb815335e;anchor=swh:1:rev:943353fd2fa6d2d2bedc80b85d788aac1930febe;path=/program_analysis/full/simplifier.ml;lines=41-81} does exactly that to simplify results and aid recurrence solving, as we will see next.  It also unrolls stubs for the sake of simplifying record projections, etc., as covered in Section \ref{sec_impl_simpl_results}.

To help avoid bad paths early on, we precompute for each function a set of variables that can be looked up to a value from that point, as described in Section \ref{sec_overview_analysis_impl}.  Concretely, it is implemented as a list due to limitations with building \texttt{Map}s/\texttt{Set}s whose keys also contain \texttt{Map}s/\texttt{Set}s. Before the analysis starts, \texttt{scope\_vars} \gh{https://archive.softwareheritage.org/swh:1:cnt:1efd49d53593fda7d2fe8fbe3593f9e9e92cf538;origin=https://github.com/JHU-PL-Lab/dde;visit=swh:1:snp:310bdd3c60f2c48c7916c12e7e70aa5eb815335e;anchor=swh:1:rev:943353fd2fa6d2d2bedc80b85d788aac1930febe;path=/interpreter/src/ast.ml;lines=287-322} parses the program expression in a top-down fashion while collecting the variables that are in scope.  Upon visiting each function, it increments the DeBruijn index of each of these variables and then adds a $0$-indexed variable for the parameter of the current function.  This new variable list is then attached to this function and propagated forward as \texttt{scope\_vars} visits the body of the function, and so on.

When looking up a non-local variable, after stitching stacks, we check if it exists in the variable list carried by each candidate function (applied in the top stack frame).  If it does, we proceed to decrement its DeBruijn index and continue the analysis with the call stack in the closure of this candidate function; otherwise, we skip this candidate function.

\subsubsection*{Solver}
\label{appendix_b_solver}

When analyzing conditionals, the implementation follows the methodology described in Section \ref{sec_eval_recur_chcs} to prune away spurious branches.  When creating a CHC solver instance, we specify Z3's ``HORN'' mode \gh{https://archive.softwareheritage.org/swh:1:cnt:c7dcdc6f22fc27fecdedec1af11014125d36c918;origin=https://github.com/JHU-PL-Lab/dde;visit=swh:1:snp:310bdd3c60f2c48c7916c12e7e70aa5eb815335e;anchor=swh:1:rev:943353fd2fa6d2d2bedc80b85d788aac1930febe;path=/program_analysis/full/lib.ml;lines=16} to enable Spacer \cite{SpacerTutorial22}.

To utilize Spacer for greater path-sensitivity and verifying result correctness against a user-provided assertion, we implement an intricate yet disciplined translation from our analysis result data structures into CHCs \gh{https://archive.softwareheritage.org/swh:1:cnt:a66e48539fea93667d677cb55eea50a4faf869a2;origin=https://github.com/JHU-PL-Lab/dde;visit=swh:1:snp:310bdd3c60f2c48c7916c12e7e70aa5eb815335e;anchor=swh:1:rev:943353fd2fa6d2d2bedc80b85d788aac1930febe;path=/program_analysis/full/solver.ml;lines=174-441}, per the rules in Figure \ref{fig_tochc}.  The mutually recursive \texttt{chcs\_of\_assert}, \texttt{chcs\_of\_cond} (generates CHCs from path conditions), \texttt{chcs\_of\_atom}, and \texttt{chcs\_of\_res} thread a state monad \gh{https://archive.softwareheritage.org/swh:1:cnt:a66e48539fea93667d677cb55eea50a4faf869a2;origin=https://github.com/JHU-PL-Lab/dde;visit=swh:1:snp:310bdd3c60f2c48c7916c12e7e70aa5eb815335e;anchor=swh:1:rev:943353fd2fa6d2d2bedc80b85d788aac1930febe;path=/program_analysis/full/solver.ml;lines=72} that carries a number of data structures that are modified over time.  These include the set of CHCs generated, IDs for value atoms and value results (sets of atoms), and a map from IDs to Z3 predicates $X_{\id(\reval)}$.  An ``entry'' Z3 assertion to the CHCs \gh{https://archive.softwareheritage.org/swh:1:cnt:a66e48539fea93667d677cb55eea50a4faf869a2;origin=https://github.com/JHU-PL-Lab/dde;visit=swh:1:snp:310bdd3c60f2c48c7916c12e7e70aa5eb815335e;anchor=swh:1:rev:943353fd2fa6d2d2bedc80b85d788aac1930febe;path=/program_analysis/full/solver.ml;lines=434-436} is recorded for constructing the additional ``guiding'' CHC assertion on \camlil!if! conditions, to help narrow down on which branch(es) to take.  This entry assertion is always the predicate with ID \texttt{P0}, which corresponds to the outermost analysis result.  \texttt{solve\_cond} \gh{https://archive.softwareheritage.org/swh:1:cnt:c7dcdc6f22fc27fecdedec1af11014125d36c918;origin=https://github.com/JHU-PL-Lab/dde;visit=swh:1:snp:310bdd3c60f2c48c7916c12e7e70aa5eb815335e;anchor=swh:1:rev:943353fd2fa6d2d2bedc80b85d788aac1930febe;path=/program_analysis/full/lib.ml;lines=15-34} and \texttt{verify\_result} \gh{https://archive.softwareheritage.org/swh:1:cnt:a66e48539fea93667d677cb55eea50a4faf869a2;origin=https://github.com/JHU-PL-Lab/dde;visit=swh:1:snp:310bdd3c60f2c48c7916c12e7e70aa5eb815335e;anchor=swh:1:rev:943353fd2fa6d2d2bedc80b85d788aac1930febe;path=/program_analysis/full/solver.ml;lines=444-453} kick off the process of generating CHCs and then solving them.

Of particular note is how we generate a CHC with \texttt{chcs\_of\_assert} \gh{https://archive.softwareheritage.org/swh:1:cnt:a66e48539fea93667d677cb55eea50a4faf869a2;origin=https://github.com/JHU-PL-Lab/dde;visit=swh:1:snp:310bdd3c60f2c48c7916c12e7e70aa5eb815335e;anchor=swh:1:rev:943353fd2fa6d2d2bedc80b85d788aac1930febe;path=/program_analysis/full/solver.ml;lines=174-213}, which handles the case of user-provided assertions on analysis results in the form of $\gtletassert\ $ \camlil!x =$$ e$$ in assn!.  \texttt{e} is an expression while \texttt{assn} is written in an assertion syntax.  This additional CHC is key to guiding Spacer to derive a more precise model.  At the current stage, we restrict the syntax of an assertion that involves a variable \texttt{x} that represents the locally bound analysis result to be \texttt{x} (for asserting \texttt{x} to be $\gttrue$), \camlil!not x!, or \texttt{x <operator> <variable-free expression>}.  Since \texttt{x} and the RHS expression serve the purpose of helping generate a CHC instead of evaluating to a value, they cannot be expressed with normal expressions.  Thus, we define a type \texttt{Res\_fv.t} \gh{https://archive.softwareheritage.org/swh:1:cnt:1efd49d53593fda7d2fe8fbe3593f9e9e92cf538;origin=https://github.com/JHU-PL-Lab/dde;visit=swh:1:snp:310bdd3c60f2c48c7916c12e7e70aa5eb815335e;anchor=swh:1:rev:943353fd2fa6d2d2bedc80b85d788aac1930febe;path=/interpreter/src/ast.ml;lines=160-170} to allow constructing result values with ``free'' variables in them.  With this in place, we first confirm that an assertion has a valid form and convert it into a \texttt{Res\_fv.t} expression \gh{https://archive.softwareheritage.org/swh:1:cnt:c7dcdc6f22fc27fecdedec1af11014125d36c918;origin=https://github.com/JHU-PL-Lab/dde;visit=swh:1:snp:310bdd3c60f2c48c7916c12e7e70aa5eb815335e;anchor=swh:1:rev:943353fd2fa6d2d2bedc80b85d788aac1930febe;path=/program_analysis/full/lib.ml;lines=82-115}, and then translate it into a CHC.  \texttt{x} in the assertion allows us to look up the Z3 predicate generated for the top layer of the analysis result that we are asserting against.

To reduce clutter and improve readability, the CHC translation algorithm benefits from a collection of convenience shorthands \gh{https://archive.softwareheritage.org/swh:1:cnt:a66e48539fea93667d677cb55eea50a4faf869a2;origin=https://github.com/JHU-PL-Lab/dde;visit=swh:1:snp:310bdd3c60f2c48c7916c12e7e70aa5eb815335e;anchor=swh:1:rev:943353fd2fa6d2d2bedc80b85d788aac1930febe;path=/program_analysis/full/solver.ml;lines=12-36} that abbreviate common Z3 data constructors.  For example, the OCaml binary operator \texttt{a -{}-> b} calls the Z3 OCaml API to construct an implication \texttt{(assert (=> a b))} (in SMT-LIB \cite{BarFT-SMTLIB} syntax).

\subsubsection*{Simplifier}

The simplifier is another important part of our program analysis, as discussed in Section \ref{sec_impl_simpl_results}, where the \texttt{simpl\_atom} and \texttt{simpl\_res} mutually recursive functions \gh{https://archive.softwareheritage.org/swh:1:cnt:b8f2241c0d86e6b67873e25273ca4cbe6d5b01cd;origin=https://github.com/JHU-PL-Lab/dde;visit=swh:1:snp:310bdd3c60f2c48c7916c12e7e70aa5eb815335e;anchor=swh:1:rev:943353fd2fa6d2d2bedc80b85d788aac1930febe;path=/program_analysis/full/simplifier.ml;lines=102-316} implement the rewrite rules.  \texttt{simpl\_res} halts when the simplified expression produced from the latest round does not change from the input.  Due to the frequent pattern matching done on results in \texttt{simplifier}, we implement value results as OCaml lists instead of sets, since the latter is opaque and cannot be pattern matched.

\subsubsection*{Visualizations}

To improve usability and help demonstrate the expressiveness of our analysis, we implement a translation \gh{https://archive.softwareheritage.org/swh:1:cnt:1a9e7ea04ab1193a660f566a4011b1eae05f2920;origin=https://github.com/JHU-PL-Lab/dde;visit=swh:1:snp:310bdd3c60f2c48c7916c12e7e70aa5eb815335e;anchor=swh:1:rev:943353fd2fa6d2d2bedc80b85d788aac1930febe;path=/program_analysis/full/graph.ml;lines=346} from analysis results into DOT programs that Graphviz \cite{Graphviz} uses to generate visualizations.  We illustrate two example programs this way in Section \ref{sec_analysis_results}.

In graphs like Figure \ref{fig_id} and \ref{fig_map}, nodes are primitive values (\eg ints and bools), operators, value results (denoted by \texttt{|}), and stubs.  Edges mostly point from operators to their operands, from value results to value atoms, and from records to their fields.  The only exception is the case of back edges from stubs to their value result parent.

This translation process shares many similarities with translating results into CHCs.  Here, we also have a mutually recursive set of functions \texttt{dot\_of\_atom} and \texttt{dot\_of\_res} and mechanisms to generate unique IDs for value atoms and value results.  These are parts of the state threaded through the translation via a monad \gh{https://archive.softwareheritage.org/swh:1:cnt:1a9e7ea04ab1193a660f566a4011b1eae05f2920;origin=https://github.com/JHU-PL-Lab/dde;visit=swh:1:snp:310bdd3c60f2c48c7916c12e7e70aa5eb815335e;anchor=swh:1:rev:943353fd2fa6d2d2bedc80b85d788aac1930febe;path=/program_analysis/full/graph.ml;lines=192}.  Since new IDs are generated for results based on their \emph{structural} equality, to ensure uniqueness, we define data types (\gh{https://archive.softwareheritage.org/swh:1:cnt:1a9e7ea04ab1193a660f566a4011b1eae05f2920;origin=https://github.com/JHU-PL-Lab/dde;visit=swh:1:snp:310bdd3c60f2c48c7916c12e7e70aa5eb815335e;anchor=swh:1:rev:943353fd2fa6d2d2bedc80b85d788aac1930febe;path=/program_analysis/full/graph.ml;lines=9-32}, \gh{https://archive.softwareheritage.org/swh:1:cnt:1a9e7ea04ab1193a660f566a4011b1eae05f2920;origin=https://github.com/JHU-PL-Lab/dde;visit=swh:1:snp:310bdd3c60f2c48c7916c12e7e70aa5eb815335e;anchor=swh:1:rev:943353fd2fa6d2d2bedc80b85d788aac1930febe;path=/program_analysis/full/graph.ml;lines=159}) parallel to \texttt{Atom.t} \gh{https://archive.softwareheritage.org/swh:1:cnt:2f8dda23d8881c12e679f1d980d00e3bced4e5e8;origin=https://github.com/JHU-PL-Lab/dde;visit=swh:1:snp:310bdd3c60f2c48c7916c12e7e70aa5eb815335e;anchor=swh:1:rev:943353fd2fa6d2d2bedc80b85d788aac1930febe;path=/program_analysis/full/utils.ml;lines=100-123} and \texttt{Res.t} \gh{https://archive.softwareheritage.org/swh:1:cnt:2f8dda23d8881c12e679f1d980d00e3bced4e5e8;origin=https://github.com/JHU-PL-Lab/dde;visit=swh:1:snp:310bdd3c60f2c48c7916c12e7e70aa5eb815335e;anchor=swh:1:rev:943353fd2fa6d2d2bedc80b85d788aac1930febe;path=/program_analysis/full/utils.ml;lines=192} that assign \emph{every} constructor a unique label so that each AST node is structurally unique.  Then, we implement \texttt{mk} functions to derive these augmented results from the originals.  Other components of the monad state include nodes, edges, and visualization properties for specific edges.

Much of the complication with this part comes from eliding nodes for singleton value results (which is very common in practice) and optionally nodes for path conditions.  In this process, many edges have to be remapped and much bookkeeping and case analysis are required.

\subsection{Program Analysis (Simplified)}

The purpose of the simplified program analysis is to help evaluate our pure demand semantics on a playing field more level with standard presentations of abstract interpretation.  Compared with the full version, it does not track cycles (and thus does not infer recurrences) and does not leverage CHC solving to gain precision at conditionals and verify final results.  Standard systems do not have these features either.  Stubbing is still performed to enforce termination.

In this simplified system, all integers are abstracted into \texttt{IntAnyAtom} \gh{https://archive.softwareheritage.org/swh:1:cnt:b31a3a3123bf8f77e3019adee109ca162e589f52;origin=https://github.com/JHU-PL-Lab/dde;visit=swh:1:snp:310bdd3c60f2c48c7916c12e7e70aa5eb815335e;anchor=swh:1:rev:943353fd2fa6d2d2bedc80b85d788aac1930febe;path=/program_analysis/simple/lib.ml;lines=30}.  Conditionals maintain only minimal precision as both branches are taken unless the condition is a primitive boolean value.  A binary operation like \texttt{Plus} is immediately reduced to a set of result atoms instead of \texttt{PlusAtom}, which does not carry out the operation on its operands.  As a result, \texttt{simplify} now plays a less important role and primarily performs record projections and inspections.

Otherwise, the essence of the rules in Figure \ref{fig_steps_analysis_extended} is retained in \texttt{analyze\_aux} \gh{https://archive.softwareheritage.org/swh:1:cnt:b31a3a3123bf8f77e3019adee109ca162e589f52;origin=https://github.com/JHU-PL-Lab/dde;visit=swh:1:snp:310bdd3c60f2c48c7916c12e7e70aa5eb815335e;anchor=swh:1:rev:943353fd2fa6d2d2bedc80b85d788aac1930febe;path=/program_analysis/simple/lib.ml;lines=21}.

\subsection{Logging}

Debugging OCaml programs can be difficult and is limited to print debugging due to insufficient debugger support.  To mitigate this while implementing such intricate systems, we extensively annotate the source code with debug messages and log them to a file on each run with the help of the Logs package \cite{Logs}.

To facilitate locating sibling proof trees that may be separated by many lines of logs and have a better sense of the calls and returns in our big-step systems, we prefix key log messages with their recursion depth.  This helps tremendously with traversing log files.

\subsection{Testing}

Unit tests are available for each of the aforementioned systems, and they are executable from the command line with a variety of optional flags (\eg \texttt{-{}-no-cache} to turn off caching and \texttt{-{}-verify} to solve the final result with Z3).  For lengthier test cases, \eg the DDPA benchmarks, we read them from external files.  Detailed instructions on working with the artifact \cite{SmithZhangArtifact24a} can be found in the README \gh{https://archive.softwareheritage.org/swh:1:cnt:dffdfda8d6e662de03596d2429ff425de7fa7387;origin=https://github.com/JHU-PL-Lab/dde;visit=swh:1:snp:310bdd3c60f2c48c7916c12e7e70aa5eb815335e;anchor=swh:1:rev:943353fd2fa6d2d2bedc80b85d788aac1930febe;path=/README.md}.

\subsubsection*{Benchmarks}

The program analysis test suite comes with a set of benchmarks as described in Section \ref{sec_analysis_results}.  We use \texttt{Core\_bench} \cite{Core_bench} to automatically schedule many runs of each test and generate a table with quite a few useful metrics (most notably Time/Run).  The specifc number of times a test is repeated depends on the ratio between the alloted quota and its runtime.

\subsubsection*{Profiling}

It is also important to understand the performance bottlenecks of our systems, which is why the artifact also comes with profiling support.  We use \texttt{landmarks} \cite{landmarks} to generate a hierarchy of functions called, which are ranked by the total time taken to execute each.  The programmer may find that in a lot of cases, CHC solving (\texttt{Lib.solve\_cond}) takes up a significant portion, if not most, of the runtime.

\subsection{Langage Guide}

The syntax of our language used across the interpreter and the program analyses is based on expression grammar of Figure \ref{fig_grammar_extended} and mostly follows OCaml syntax conventions.  There are a few small differences  which we now outline.

\begin{itemize}
  \item Curried functions cannot be inlined via the shorthand \camlil!fun a b c $\gtarrow$ ...! and instead have to be laid out in full via \camlil!fun a $\gtarrow$ fun b $\gtarrow$ fun c $\gtarrow$ ...!.

  \item Recursive functions have to be defined in the self-passing style (or via Y combinators), \eg \camlil!(fun self $\gtarrow$ fun x $\gtarrow$ ...) (fun self $\gtarrow$ fun x $\gtarrow$ ...) 0!.

  \item Records can be defined immediately without an established type, and inspecting a record field can be done via the $\gtin$ keyword, \eg \camlil!l in {$$ l$$ =$$ 1 $$}!.

  \item Due to the $\gtin$ keyword being reused for both record inspections and let bindings, the programmer must insert parentheses around let definitions that end in a variable, in order to disambiguate from record inspections, \eg \camlil!let a $$= $$(fun x $\gtarrow$ x) in a!.
  
  \item The syntax includes the added $\gtletassert$ expression, $\gtletassert\ $ \camlil!x =$$ e$$ in assn!, as described in Section \ref{appendix_b_solver} and at the \texttt{eval\_assert} function \gh{https://archive.softwareheritage.org/swh:1:cnt:c7dcdc6f22fc27fecdedec1af11014125d36c918;origin=https://github.com/JHU-PL-Lab/dde;visit=swh:1:snp:310bdd3c60f2c48c7916c12e7e70aa5eb815335e;anchor=swh:1:rev:943353fd2fa6d2d2bedc80b85d788aac1930febe;path=/program_analysis/full/lib.ml;lines=75-81}.
\end{itemize}

}

\end{document}